\documentclass[12pt]{article}
\usepackage[margin=1.0in]{geometry}
\usepackage{amsmath,amssymb,amsthm}
\usepackage[T1]{fontenc}
\usepackage{lmodern}
\usepackage{microtype, centernot}
\usepackage[authoryear,round]{natbib}
\usepackage{booktabs}
\usepackage{threeparttable}
\usepackage[colorlinks=true,linkcolor=blue,citecolor=blue,urlcolor=blue]{hyperref}

\newcommand{\R}{\mathbb{R}}

\theoremstyle{definition}
\newtheorem{definition}{Definition}
\theoremstyle{plain}
\newtheorem{lemma}{Lemma}
\newtheorem{proposition}{Proposition}
\newtheorem{theorem}{Theorem}
\newtheorem{corollary}{Corollary}
\newtheorem{remark}{Remark}
\begin{document}

\title{A Social Welfare Function Satisfying Anonymity and Almost Weak Pareto%
\footnote{We thank Francesco Ruscitti for several fruitful discussions on the subject of this paper.
We would also like to record our personal appreciation of his contribution: had he not expressly preferred otherwise, we would have been very pleased to have his name appear among the authors rather than in the acknowledgments.}} \author{Ram Sewak Dubey\thanks{Economics Department, Feliciano School of Business, Montclair State University, Montclair, NJ 07043, USA; E-mail: dubeyr@montclair.edu} \and Giorgio Laguzzi\thanks{Universit\`a del Piemonte Orientale, 15121 Alessandria, Italy; Email: giorgio.laguzzi@uniupo.it}} \date{\today }
\maketitle

\begin{abstract}
We study the existence of real-valued social welfare functions (SWFs) on the set of infinite utility streams that satisfy anonymity and almost weak Pareto. 
We characterize the domains of one-period utilities, \(Y\subset \mathbb{R}\), for which such SWFs exist. 
We show that an SWF satisfying these two axioms exists if and only if \(Y\) contains no subset that is order-isomorphic to the set of negative and positive integers. 
Thus, the restrictions on \(Y\) required for the existence of an SWF satisfying anonymity and almost weak Pareto coincide with those required for an SWF satisfying anonymity and weak Pareto. 
Moreover, the SWF we construct is invariant under arbitrary permutations and satisfies stationarity, two properties that are particularly useful for decision making in infinite-horizon economies.

\noindent \emph{Keywords:} \texttt{Anonymity,}\; \texttt{Almost weak Pareto,} 
\texttt{Social welfare function,}\; \texttt{Stationarity.}

\noindent \emph{Journal of Economic Literature} Classification Numbers: \texttt{D60,}\; \texttt{D70,}\; \texttt{D90.}
\end{abstract}

\newpage
\section{Introduction}

One of the toughest challenges facing a decision-maker whose task is to evaluate infinite utility streams is the existence of a trade-off between `equity' and `efficiency'.
Ideally, a decision-maker would like to rely on social preferences (or, better yet, a `utility function') that embody both principles.
For simplicity, let us refer to such preferences (utility functions) as ethical social welfare orders (SWO, henceforth) (social welfare functions, SWF henceforth).
Regarding the principle of `equity' (inter-generational equity, to be more precise), the literature has largely focused on the concept of anonymity (AN, henceforth), in its two main forms: finite anonymity and strong anonymity.
As to the principle of `efficiency', a wide spectrum of Pareto efficiency concepts suitable for infinite sequences of real numbers has been introduced and analyzed, ranging from strong Pareto to weaker forms such as weak Pareto (WP, henceforth).

We refer to the set of utility values which the single individual or generation can attain as the feasible set of one-period utilities, denoted by $Y$.
The literature shows that possibility and impossibility results depend both on the strength of the Pareto axiom and on the structure of the set $Y$.
For relatively strong efficiency requirements, representability is difficult to achieve.
In particular, \citet{dubey2020} showed that there does not exist a SWF satisfying anonymity and upper asymptotic Pareto on any non-trivial domain, while \citet{petri2019} proved the existence of a SWF satisfying anonymity and lower asymptotic Pareto when $Y$ is finite.
More recently, \citet{dubey2021} showed that a SWF satisfying anonymity and asymptotic density-one Pareto (ADOP, henceforth) exists if and only if $Y$ is finite.
At the other end, \citet{dubey2011} proved that there exists a SWF satisfying WP and AN if and only if the feasible set of one-period utilities $Y$ does not contain a subset that is order-isomorphic to the set of integers.
Thus, for WP the crucial feature of $Y$ is its order structure rather than its cardinality alone.

The above results leave a natural gap between ADOP and WP.
Since ADOP is a stronger axiom than WP, an open question of interest that arises and that we tackle is the following: is it possible to strengthen, in an economically meaningful manner, the WP axiom to an efficiency axiom weaker than ADOP, while preserving the existence of a SWF satisfying AN and efficiency?
It would also be useful to have an easy-to-check criterion for deciding possibility results with respect to the problem of aggregating infinite utility streams, preferably in terms of properties of $Y(<)$.
It is worth pointing out that \citet{dubey2021}, in their concluding remarks, raised the question of the representability and construction of social welfare relations satisfying anonymity and an ``almost weak Pareto'' condition, according to which all but finitely many generations must be better off in one utility stream than in another.
The present paper takes up this question.

With these goals in mind, in this paper we formally introduce the Almost Weak Pareto (AWP) axiom and study the existence and construction of social welfare functions satisfying AWP and AN.
The AWP axiom requires that one utility stream be socially preferred to another whenever it is no worse in every period and strictly better in all but finitely many periods.
Thus AWP is weaker than ADOP and stronger than WP.
We prove (see Corollary \ref{C1}) that there exists a SWF satisfying AWP and AN if and only if $Y$ does not contain a subset that is order-isomorphic to the set of integers.
Consequently, strengthening WP to AWP does not require any further restriction on the feasible set of one-period utilities: the domain restrictions on the set $Y$ for the existence of a SWF satisfying anonymity and AWP coincide with those needed for the existence of a SWF satisfying anonymity and WP.
Moreover, our proof of existence is constructive (see Lemma \ref{L9} and the proof of Theorem \ref{T1}), and the SWF we construct is invariant under arbitrary permutations and satisfies stationarity in addition to AWP.
Stationarity is a useful property in the context of decision making in infinite-horizon economies since it guarantees time consistency.

If $Y$ contains a subset that is order-isomorphic to the set of integers, Corollary \ref{C1} implies that no SWF satisfying AN and AWP exists.
However, since the strong Pareto axiom is stronger than AWP, it follows from \citet{svensson1980} that social welfare orders satisfying AN and AWP may still exist even when they do not admit a real-valued representation.
This raises the question whether such social welfare orders can be constructive.
Taking advantage of a result due to \citet{dubey2011a} and of the fact that AWP is stronger than WP, we show (see Corollary \ref{C2}) that if $Y$ contains a subset that is order-isomorphic to the set of integers, then every SWO satisfying AN and AWP is non-constructive.
Consequently, we establish (see Theorem \ref{T2}) that there does not exist any SWF satisfying AN and AWP if and only if every SWO on $Y^{\mathbb{N}}$ satisfying AN and AWP is non-constructive.

The methodology used for proving the existence of the aforementioned SWF also deserves mention.
Our construction is related to the order-separability approach used by \citet{debreu1954} in obtaining a numerical representation of a preference ordering.
More specifically, we construct a countable order-separating family for the relevant profile order and use it to obtain a real-valued representation suitable for the problem under consideration.

The remainder of the paper is organized as follows.
Section \ref{sec:2} gathers some preliminary concepts and notation, and introduces the equity and efficiency axioms imposed on the social welfare orders under examination.
In Section \ref{sec:3}, we state and prove our main results on the representation of social welfare orders satisfying anonymity and almost weak Pareto.
The proofs of Lemmas \ref{L5}, \ref{L6}, and \ref{L9} are collected in Appendix \ref{app}.
Section \ref{sec:4} concludes.

\section{Preliminaries and definitions}

\label{sec:2}

Let $\mathbb{R}$, $\mathbb{N}$, and $\mathbb{Z}$ denote the sets of real numbers, natural numbers, and integers respectively.
Let $\mathbb{N}_{0}=\{0\}\cup\mathbb{N}$.
Let $\mathcal{A}=\mathbb{N}_{0}\cup\{\infty\}$, ordered by $0<1<2<\cdots<\infty$.
For $x,y\in\mathbb{R}^{\mathbb{N}}$, write $x\ge y$ if $x_t\ge y_t$ for every $t\in\mathbb{N}$.
Write $x>y$ if $x\ge y$ and $x\ne y$.

Let $Y\subset\mathbb{R}$ be the set of possible utilities.
Throughout the paper, whenever $Y$ is a subset of $\mathbb{R}$, it is understood to be endowed with the usual $<$ relation inherited from $\mathbb{R}$.
Let $X=Y^{\mathbb{N}}$ be the set of feasible utility streams.
For $x\in X$, write $x=(x_1,x_2,\cdots)$.
For each $t\in\mathbb{N}$, the number $x_t\in Y$ is the utility of generation $t$.

A SWO is a binary relation $\succcurlyeq$ on $X$ which is complete and transitive.
We write $x\sim y$ if both $x\succcurlyeq y$ and $y\succcurlyeq x $ hold; and $x\succ y$ when $x\succcurlyeq y$ and $y\not\succcurlyeq x$ hold.
A social welfare function is a map $W: X\to\mathbb{R}$.
A permutation $\pi:\mathbb{N}\to\mathbb{N}$ is a bijection on $\mathbb{N}$.
A permutation $\pi$ is finite if there is $T\in\mathbb{N}$ such that $\pi(t)=t$ for every $t>T$.
For $x\in X$, write $x_{\pi}$ for the stream defined by $(x_{\pi})_t=x_{\pi(t)}$.

\subsection{Domain Types}

In this subsection, we recall some terminology concerning order types and state a useful domain condition.
A set $S$ is \emph{strictly ordered} by a binary relation $<$ if $<$ is connected, transitive, and irreflexive: for all $s, s^{\prime}\in S$ with $s \neq s^{\prime}$, either $s<s^{\prime}$ or $s^{\prime}<s$ holds; if $s<s^{\prime}$ and $s^{\prime}<s^{\prime\prime}$, then $s<s^{\prime\prime}$; and $s<s$ holds for no $s\in S$.
We denote such a strictly ordered set by $S(<)$.

Two strictly ordered sets $S(<)$ and $S^{\prime}(\prec)$ are \emph{similar} if there exists a one-to-one map $f:S\to S^{\prime}$ onto $S^{\prime}$ such that, for all $s_1,s_2\in S$, $s_1<s_2$ if and only if $f(s_1)\prec f(s_2)$.
Equivalently, the two ordered sets are order-isomorphic.

\begin{definition}
\label{D1}
Let $(Y,<)$ be a strictly ordered set.
A \emph{cut} $[Y_1,Y_2]$ of $Y$ is a partition of $Y$ into two non-empty sets $Y_1$ and $Y_2$ such that
\begin{equation*}
Y_1\cup Y_2=Y,\qquad Y_1\cap Y_2=\varnothing,
\end{equation*}
and, for every $y_1\in Y_1$ and every $y_2\in Y_2$,
\begin{equation*}
y_1<y_2.
\end{equation*}
An element $\underline{y}\in Y$ is called a \emph{first element} of $Y$ if there is no $y\in Y$ such that $y<\underline{y}$.
An element $\overline{y}\in Y$ is called a \emph{last element} of $Y$ if there is no $y\in Y$ such that $\overline{y}<y$.
\end{definition}

We say that $(Y,<)$ is \emph{well ordered} if every non-empty subset of $Y$ has a first element, and \emph{reverse well ordered} if every non-empty subset of $Y$ has a last element.
If $Y_1$ and $Y_2$ are disjoint ordered subsets whose union is $Y$ and every element of $Y_1$ is smaller than every element of $Y_2$, we write $Y=Y_1+Y_2$ for their ordered sum; either component may be empty.

Let $Y$ be a non-empty subset of $\mathbb{R}$.
Let $\mathbb{N}^{-}=\{-1,-2,-3,\cdots\}$, ordered by the usual order.
We use the following order types:

\begin{itemize}
\item $Y(<)$ is of order type $\omega$ if $Y(<)$ is similar to $\mathbb{N}(<)$.

\item $Y(<)$ is of order type $\omega^{\ast}$ if $Y(<)$ is similar to $\mathbb{N}^{-}(<)$.

\item $Y(<)$ is of order type $\omega+\omega^{\ast}$ if there is a cut $[Y_1,Y_2]$ of $Y$ such that $Y_1(<)$ is of order type $\omega$ and $Y_2(<)$ is of order type $\omega^{\ast}$.
\end{itemize}

The domain restriction used in our characterization can be stated directly in terms of the set of integers.
The following elementary observation will be useful.

\begin{lemma}
\label{L1}
\emph{Let $(Y,<)$ be a strictly ordered set.
Every non-empty subset of $Y$ has either a first element or a last element if and only if $Y$ does not contain a subset that is order-isomorphic to $\mathbb{Z}$.}
\end{lemma}

\begin{proof}
Suppose first that $Y$ contains a subset that is order-isomorphic to $\mathbb{Z}$.
Such a subset has neither a first element nor a last element.
Hence it is not true that every non-empty subset of $Y$ has either a first or a last element.

Conversely, suppose that some non-empty subset $A\subseteq Y$ has neither a first element nor a last element.
Choose $a_0\in A$.
Since $A$ has no last element, choose successively $a_1,a_2,\cdots$ in $A$ such that $a_0<a_1<a_2<\cdots$.
Since $A$ has no first element, choose successively $a_{-1},a_{-2},\cdots$ in $A$ such that $\cdots<a_{-2}<a_{-1}<a_0$.
Then the subset $\{a_n:n\in\mathbb{Z}\}$ is order-isomorphic to $\mathbb{Z}$.
\end{proof}

The following result in \citet[Theorem 1]{Steckel1928} will be useful for our characterization theorem.
\begin{lemma}
\label{L2}
\emph{In order that every nonempty subset of an ordered set $(Y, <)$ possess either a first element or a last element, it is necessary and sufficient that $(Y, <)$ can be represented as the ordered sum of two disjoint ordered sets
\begin{equation*}
Y=Y_1+Y_2,
\end{equation*}
where $(Y_1, <)$ is well ordered (possibly empty) and $(Y_2, <)$ is reverse well ordered (possibly empty).}
\end{lemma}

\subsection{Equity and Efficiency Axioms}

\begin{definition}[Finite anonymity]
\label{D2}
A social welfare function $W:X\to\mathbb{R}$ satisfies finite anonymity (hereinafter referred to as AN) if for every $x\in X$ and every finite permutation $\pi$ of $\mathbb{N}$,
\begin{equation*}
W(x_{\pi})=W(x).
\end{equation*}
\end{definition}

\begin{definition}[Strong anonymity]
\label{D3}
The SWF $W$ satisfies strong anonymity if, for every $x\in X$ and every permutation $\pi$ of $\mathbb{N}$,
\begin{equation*}
W(x_{\pi})=W(x).
\end{equation*}
\end{definition}

\begin{definition}[Almost Weak Pareto]
\label{D4}
A social welfare function $W:X\to\mathbb{R}$ satisfies Almost Weak Pareto (hereinafter referred to as AWP) if for all $x,y\in X$, whenever $x\ge y$ and there exists $\bar N\in\mathbb{N}$ such that
\begin{equation*}
x_t>y_t \quad\text{for every }t>\bar N,
\end{equation*}
then
\begin{equation*}
W(x)>W(y).
\end{equation*}
\end{definition}

\begin{definition}[Weak Pareto]
\label{D5}
A social welfare function $W:X\to\mathbb{R}$ satisfies Weak Pareto (hereinafter referred to as WP) if for all $x,y\in X$,
\begin{equation*}
x_t>y_t \quad\text{for every }t\in \mathbb{N},
\end{equation*}
then
\begin{equation*}
W(x)>W(y).
\end{equation*}
\end{definition}

Observe that AWP is stronger than WP.
Hence, if there is no SWF satisfying WP then there is no SWF satisfying AWP as well.
When AN, AWP, or WP is imposed directly on a social welfare order $\succcurlyeq$, we use the corresponding preference formulation: finite permutations imply indifference, and the Pareto antecedents imply the strict preference $x\succ y$.

\begin{definition}[Monotonicity]
\label{D6}
A social welfare function $W:X\to\mathbb{R}$ is monotone if, for all $x,y\in X$,
\begin{equation*}
x\ge y \quad\Longrightarrow\quad W(x)\ge W(y).
\end{equation*}
\end{definition}

\begin{definition}[Stationarity]
\label{D7}
A social welfare function $W:X\to\mathbb{R}$ satisfies stationarity if, for all $x,y\in X$ with $x_1=y_1$,
\begin{equation*}
W(x)\ge W(y) \quad\Longleftrightarrow\quad W(x^{\prime})\ge W(y^{\prime}),
\end{equation*}
where
\begin{equation*}
x^{\prime}=(x_2,x_3,\cdots), \quad y^{\prime}=(y_2,y_3,\cdots).
\end{equation*}
\end{definition}

\section{SWF satisfying Anonymity, Almost Weak Pareto, and Stationarity}

\label{sec:3}
In the subsections below, we develop the existence result in three stages.
We first construct a SWF when $Y\subset\mathbb{R}$ is of order type $\omega$, and then extend the construction to an arbitrary well ordered subset of $\mathbb{R}$.
We next construct a SWF when $Y\subset\mathbb{R}$ is of order type $\omega^{\ast}$, and then extend the result to an arbitrary reverse well ordered subset of $\mathbb{R}$.
Finally, we construct a SWF when $Y\subset\mathbb{R}$ is of order type $\omega+\omega^{\ast}$, and then extend the construction, using Lemma \ref{L2}, to an arbitrary subset of $\mathbb{R}$ which can be represented as the ordered sum of a well ordered set and a reverse well ordered set.
In all these cases, the ordering on the subsets of $\mathbb{R}$ is the usual $<$ relation.

\subsection{SWF for $Y$ of order type $\protect\omega$}

We construct a social welfare function that satisfies AN, AWP, and stationarity when the domain $Y\subset\mathbb{R}$ is of order type $\omega$.
To demonstrate the construction cleanly, we take $Y=\mathbb{N}=\{1,2,3,\cdots\}$.
For $x\in X=Y^{\mathbb{N}}$ and $i\in\mathbb{N}$, define the multiplicity of the utility level $i\in \mathbb{N}$ in $x$ by
\begin{equation*}
m_i(x)=\left| \{t\in\mathbb{N}:x_t= i\}\right|\in\mathcal{A},
\end{equation*}
where $\left| S\right|$ denotes the cardinality of the set $S$.
The cutoff counts
\begin{equation}
\label{E1}
C_r(x)=\left| \{t\in\mathbb{N}:x_t< r\}\right|,\quad r\in\mathbb{N},
\end{equation}
are related to these multiplicities by
\begin{equation}
\label{E2}
C_r(x)=m_1(x)+\cdots+m_{r-1}(x),
\end{equation}
with the convention that an empty sum is zero.
Thus the first-difference comparison of the cutoff counts is equivalently described by first comparing the multiplicities of $1$, then the multiplicities of $2$, and so on.
For computing the first difference, we follow the convention that $\infty-n=\infty$.

For each $x\in X$, define its canonical multiplicity profile $p(x)$ as follows.
\begin{equation}
\label{E3}
p(x)=
\begin{cases}
\left(m_1(x),\cdots,m_{j-1}(x),\infty,\infty,\cdots\right), & \text{if}\;
\exists\; \text{a first index}\; j\; \text{such that}\; m_j(x)=\infty, \\
\left(m_1(x),m_2(x),m_3(x),\cdots\right), & \text{if}\; m_i(x)<\infty\;
\text{for every}\; i\in\mathbb{N}.\end{cases}\end{equation}
In the second case of \eqref{E3}, since $x$ is an infinite stream, the finite multiplicities satisfy
\begin{equation*}
\sum_{i=1}^{\infty}m_i(x)=\infty.
\end{equation*}
Let
\begin{equation*}
\mathcal{P}=\{p(x):x\in X\}
\end{equation*}
be the set of all canonical multiplicity profiles.
Define a lexicographic binary relation on $\mathcal{P}$ as follows.
For $p,q\in\mathcal{P}$, say that $p$ is at least as good as $q$, and write $p\succeq q$, if either $p=q$, or, letting
\begin{equation*}
k^{\ast}=\min\{k\in\mathbb{N}:p_k\ne q_k\},\; \text{one has}\;
p_{k^{\ast}}<q_{k^{\ast}},
\end{equation*}
where every finite integer is taken to be $<\infty$.
Thus the profile with fewer occurrences of the lowest utility level is ranked higher; if those multiplicities are equal, the profile with fewer occurrences of the second-lowest utility level is ranked higher; and so on.\footnote{
The lexicographic relation $\succeq$ on $\mathcal{P}$ is complete and transitive, since for any two distinct profiles the set of coordinates at which they differ has a least element, and the comparison is made at that coordinate.}
The following lemma proves the existence of a countable order-separating family of integer sequences.

\begin{lemma}
\label{L3}

\emph{Let
\begin{equation*}
\mathcal{D}=\{(a_1,\cdots,a_{j-1},\infty,\infty,\cdots):j\in\mathbb{N}, \
a_1,\cdots,a_{j-1}\in\mathbb{N}_{0}\}.
\end{equation*}
Then $\mathcal{D}$ is countable and order-separating in the following sense: if $p,q\in\mathcal{P}$ and $p\succ q$, then there exists $d\in\mathcal{D}$ such that
\begin{equation*}
p\succeq d\succ q.
\end{equation*}}
\end{lemma}

\begin{proof}
The set $\mathcal{D}$ is countable because each of its elements is determined by a finite sequence of nonnegative integers.
Equivalently, after identifying a finite block $(a_1,\cdots,a_m)$ with $(a_1,\cdots,a_m,\infty,\infty,\cdots)$, one may write
\begin{equation*}
\mathcal{D}=\bigcup_{m=0}^{\infty}\mathbb{N}_{0}^m.
\end{equation*}
Each set $\mathbb{N}_{0}^m$ is countable, and the countable union of countable sets is countable.
This is the standard denumerability fact used, for example, in the theorem of \citet[Chapter III, Section 4, Theorem 1]{sierpinski1965} that the set of all finite sequences of rational numbers is effectively denumerable.
Hence one may fix an enumeration
\begin{equation*}
\mathcal{D}=\{d^1,d^2,d^3,\cdots\}.
\end{equation*}
Now let $p,q\in\mathcal{P}$ satisfy $p\succ q$.
Let
\begin{equation*}
k=\min\{i\in\mathbb{N}:p_i\ne q_i\}.
\end{equation*}
Then $p_i=q_i$ for every $i<k$ and $p_k<q_k$.
There are two cases.

\begin{enumerate}
\item Suppose $q_k\ge p_k+2$.
Define
\begin{equation*}
d=(p_1, \cdots, p_{k-1}, p_k+1,\infty,\infty,\cdots).
\end{equation*}
Then $d\in\mathcal{D}$.
The first difference between $p$ and $d$ occurs at $k $, and $p_k<p_k+1$, so $p\succ d$.
The first difference between $d$ and $q$ also occurs at $k$.
If $q_k=\infty$, then $p_k+1<q_k$ by the convention that every finite integer is smaller than $\infty$; if $q_k$ is finite, then $p_k+1<q_k$ because $q_k\ge p_k+2$.
Hence $d\succ q$.
Therefore $p\succeq d\succ q$.

\item Suppose $q_k=p_k+1$.
There are two possibilities.

\begin{itemize}
\item If $p\in\mathcal{D}$, choose $d=p$.
Then $p\succeq d$, and $d\succ q$ because $p\succ q$.

\item If $p\notin\mathcal{D}$, then, since $p\in\mathcal{P}$, every coordinate of $p$ is finite and
\begin{equation*}
\sum_{i=1}^{\infty}p_i=\infty.
\end{equation*}
Choose any $\ell>k$ and define
\begin{equation*}
d=(p_1,\cdots,p_{\ell-1},p_{\ell}+1,\infty,\infty,\cdots).
\end{equation*}
Then $d\in\mathcal{D}$.
The first difference between $p$ and $d$ occurs at $\ell$, and $p_{\ell}<p_{\ell}+1$, so $p\succ d$.
The first difference between $d$ and $q$ occurs at $k$, and $d_k=p_k<q_k$, so $d\succ q$.
Again $p\succeq d\succ q$.
\end{itemize}
\end{enumerate}

Thus $\mathcal{D}$ is order-separating.
\end{proof}
In the next lemma we establish a real-valued representation of the lexicographic order.

\begin{lemma}
\label{L4}
\emph{There exists a function $V:\mathcal{P}\to\mathbb{R}$ such that, for all $p,q\in\mathcal{P}$,
\begin{equation*}
V(p)\ge V(q) \quad\text{if and only if}\quad p\succeq q.
\end{equation*}}
\end{lemma}

\begin{proof}

Enumerate the countable set $\mathcal{D}$ as
\begin{equation*}
\mathcal{D}=\{d^1,d^2,d^3,\cdots\}.
\end{equation*}
For $p\in\mathcal{P}$, define
\begin{equation*}
V(p)=\sum_{s=1}^{\infty}\frac{1}{2^s}\mathbf{1}_{\{p\succeq d^s\}}.
\end{equation*}
This sum is well-defined and lies in $[0,1]$.

If $p\succeq q$, then for every $d^s\in\mathcal{D}$, the implication
\begin{equation*}
q\succeq d^s \quad\Longrightarrow\quad p\succeq d^s
\end{equation*}
holds.
Hence
\begin{equation*}
\mathbf{1}_{\{p\succeq d^s\}}\ge \mathbf{1}_{\{q\succeq d^s\}}
\end{equation*}
for every $s$, and therefore $V(p)\ge V(q)$.

If $p\succ q$, then, by the preceding lemma, there exists $d^{s^{\ast}}\in\mathcal{D}$ such that
\begin{equation*}
p\succeq d^{s^{\ast}}\succ q.
\end{equation*}
For this index ${s^{\ast}}$,
\begin{equation*}
\mathbf{1}_{\{p\succeq d^{s^{\ast}}\}}=1 \quad\text{and}\quad \mathbf{1}_{\{q\succeq d^{s^{\ast}}\}}=0.
\end{equation*}
Moreover,
\begin{equation*}
\mathbf{1}_{\{q\succeq d^s\}}\le\mathbf{1}_{\{p\succeq d^s\}}
\end{equation*}
for every $s$, with strict inequality at $s^{\ast}$.
Therefore $V(p)>V(q)$.

Thus $p\succ q$ implies $V(p)>V(q)$, while $p=q$ implies $V(p)=V(q)$.
Since the order on $\mathcal{P}$ is complete, $V(p)\ge V(q)$ holds if and only if $p\succeq q$.
\end{proof}

Using the representation established in the preceding lemma, we now construct a social welfare function for the case $Y=\mathbb{N}$.

\begin{proposition}
\label{P1}
\emph{Let $Y=\mathbb{N}$ and $X=Y^{\mathbb{N}}$.
Then there exists a monotone social welfare function $W:X\to\mathbb{R}$ satisfying Anonymity, Almost Weak Pareto, and stationarity.}
\end{proposition}

\begin{proof}
Let $V:\mathcal{P}\to\mathbb{R}$ be the representation constructed in the preceding lemma.
Define
\begin{equation}
\label{E4}
W(x)=V(p(x)) \quad\text{for every }x\in X.
\end{equation}
We verify the asserted properties.

\begin{itemize}
\item \textbf{Anonymity (AN).}
A finite permutation preserves the multiplicity of every utility level.
Thus, for every finite permutation $\pi$ of $\mathbb{N}$ and every $i\in\mathbb{N}$,
\begin{equation*}
m_i(x_{\pi})=m_i(x).
\end{equation*}
Therefore $p(x_{\pi})=p(x)$ and, by \eqref{E4},
\begin{equation*}
W(x_{\pi})=W(x).
\end{equation*}
So $W$ satisfies finite anonymity.
Observe that even in the case of arbitrary permutations, the multiplicity of every utility level is preserved.
Hence $W$ satisfies the strong anonymity axiom as well.

\item \textbf{Stationarity.}
Take $x,y\in X$ with $x_1=y_1$.
Let
\begin{equation*}
x^{\prime}=(x_2,x_3,\cdots), \quad y^{\prime}=(y_2,y_3,\cdots).
\end{equation*}
Suppose $x_1=y_1=r$ for some $r\in\mathbb{N}$.
Deleting the common first coordinate changes only the multiplicity of the level $r$.
For every $i\ne r$, one has
\begin{equation*}
m_i(x^{\prime})=m_i(x), \quad m_i(y^{\prime})=m_i(y).
\end{equation*}
At level $r$, deleting the common first coordinate replaces each finite positive value $a$ by $a-1$ and leaves $\infty$ equal to $\infty$.
This operation preserves equality and strict inequality on
\begin{equation*}
\{0,1,2,\cdots\}\cup\{\infty\}.
\end{equation*}
Moreover, by the definition of the canonical profile in \eqref{E3}, the first level at which an infinite multiplicity appears is unchanged by deleting a common coordinate.
Hence the canonical profiles are compared in the same way before and after deletion:
\begin{equation*}
p(x)\succeq p(y) \quad\Longleftrightarrow\quad p(x^{\prime})\succeq
p(y^{\prime}).
\end{equation*}
By Lemma \ref{L4} and \eqref{E4}, therefore,
\begin{equation*}
W(x)\ge W(y) \quad\Longleftrightarrow\quad W(x^{\prime})\ge W(y^{\prime}).
\end{equation*}
Thus $W$ satisfies stationarity.

\item \textbf{Monotonicity.}
Let $x\ge y$, i.e., $x_t\ge y_t$ for every $t\in\mathbb{N}$.
Hence, for every cutoff $r\in\mathbb{N}$,
\begin{equation*}
\{t:x_t< r\}\subseteq \{t:y_t< r\}.
\end{equation*}
Consequently,
\begin{equation}
\label{E5}
C_r(x)\le C_r(y) \quad\text{for every }r\in\mathbb{N}.
\end{equation}
If $p(x)=p(y)$, then $W(x)=W(y)$.
Otherwise let $k$ be the first index such that $p_k(x)\neq p_k(y)$.
Since all earlier multiplicities agree, the first cutoff-count difference occurs at $r=k+1$, and inequality \eqref{E5} implies
\begin{equation*}
p_k(x)<p_k(y).
\end{equation*}
Observe that the strict inequality holds for the reasons that

\begin{itemize}
\item if $p_i(x)=p_i(y)=\infty$ for some $i<k$, both profiles are truncated to $\infty$ from $i$ onward, contradicting a later first difference at $k$.
So in fact $p_i(x)=p_i(y)=m_i(x)=m_i(y)<\infty$ i.e., finite for $i<k$,

\item for $C_{k+1}$ to be the first strict-cutoff index, $C_{k+1}(x)$ must be finite, which forces $p_k(x)=m_k(x)$ to be finite.
Two cases arise: (i) Either $p_k(y)$ is finite - then $p_k(x)<p_k(y)$, or, (ii) $p_k(y) =\infty$ - then $p_k(x)<\infty=p_k(y)$.
\end{itemize}

Thus $p(x)\succ p(y)$ and hence $W(x)>W(y)$.
Therefore $W(x)\ge W(y)$, and $W $ is monotone.

\item \textbf{Almost Weak Pareto (AWP).}
Suppose $x_t\ge y_t$ for every $t\in\mathbb{N}$ and that there exists a finite integer $\bar N\in\mathbb{N}$ such that
\begin{equation*}
x_t>y_t\quad\text{for all }t>\bar N.
\end{equation*}
Fix one such $\bar N$.
As in the argument leading to \eqref{E5}, the same inequality holds for every cutoff $r\in\mathbb{N}$.
We next show that the inequality is strict for at least one cutoff.
Since $Y=\mathbb{N}$, the nonempty set
\begin{equation*}
\{x_t:t>\bar N\}
\end{equation*}
has a least element.
Let
\begin{equation*}
r_{\ast} = \min_{t>\bar N} x_t.
\end{equation*}
Because $x_t>y_t\ge 1$ for every $t>\bar N$, we have $r_{\ast}\ge 2$.
Choose $t_{\ast}>\bar N$ such that $x_{t_{\ast}}=r_{\ast}$.
Then
\begin{equation*}
y_{t_{\ast}}<r_{\ast} \quad\text{and}\quad x_{t_{\ast}}\not<r_{\ast}.
\end{equation*}
Moreover, for every $s>\bar N$, $x_s\ge r_{\ast}$.
Hence
\begin{equation*}
\{s:x_s<r_{\ast}\}\subseteq \{1,\cdots,\bar N\},
\end{equation*}
so this set is finite.
Since $t_{\ast}\in\{s:y_s<r_{\ast}\}$ but $t_{\ast}\notin\{s:x_s<r_{\ast}\}$, the definition in \eqref{E1} gives
\begin{equation*}
C_{r_{\ast}}(x)<C_{r_{\ast}}(y).
\end{equation*}
Thus the set
\begin{equation*}
\{r\in\mathbb{N}:C_r(x)<C_r(y)\}
\end{equation*}
is nonempty.
Let
\begin{equation*}
r_0=\min\{r\in\mathbb{N}:C_r(x)<C_r(y)\}.
\end{equation*}
Since $C_1(x)=C_1(y)=0$, we have $r_0\ge 2$.
Set $k=r_0-1$.
For every $r<r_0$, $C_r(x)=C_r(y)$, so in particular
\begin{equation*}
C_k(x)=C_k(y).
\end{equation*}
Also, because $C_{k+1}(x)<C_{k+1}(y)$, $C_{k+1}(x)$ is finite, and hence $C_k(x)=C_k(y)$ is finite.
Therefore, using \eqref{E2},
\begin{equation*}
m_k(x) = C_{k+1}(x)-C_k(x) < C_{k+1}(y)-C_k(y) = m_k(y).
\end{equation*}
Observe that $C_{k+1}(y)$ could be either finite or infinite, whereas $C_k(y) $ is finite.
In case $C_{k+1}(y)$ is finite, the strict inequality holds directly and in case $C_{k+1}(y)$ is infinite, the strict inequality holds because $\infty-n=\infty$ for any finite $n$.
Hence $p_k(x)<p_k(y)$, and $k$ is the first coordinate at which $p(x)$ and $p(y)$ differ.
Therefore $p(x)\succ p(y)$.
By Lemma \ref{L4}, we get
\begin{equation*}
W(x)=V(p(x))>V(p(y))=W(y).
\end{equation*}
Thus $W$ satisfies AWP.
\end{itemize}
\end{proof}

The following lemma extends the possibility result from order type $\omega$ to an arbitrary non-empty well ordered subset of $\mathbb{R}$, ordered by the usual $<$ relation.

\begin{lemma}
\label{L5}
\emph{Let $Y\subset \R$ be an arbitrary non-empty well ordered set with respect to the usual $<$ relation, and let $X=Y^{\mathbb{N}}$.
Then there exists a monotone social welfare function $W:X\to\mathbb{R}$ satisfying Anonymity, Almost Weak Pareto, and stationarity.}
\end{lemma}
\begin{proof}
See Appendix \ref{app}.
\end{proof}

\subsection{SWF for $Y$ of order type $\protect\omega^{\ast}$}

The preceding construction also yields the reverse well ordered case.
The simplest example is
\begin{equation*}
Y^- = \{-1,-2,-3,\cdots\},
\end{equation*}
with the usual order inherited from $\mathbb{R}$.
More generally, suppose that $Y\subset\mathbb{R}$ has order type $\omega^{\ast}$.
Under the order type $\omega^{\ast}$ assumption, there exists an order-reversing bijection
\begin{equation*}
\tau:Y\to\mathbb{N}
\end{equation*}
such that, for all $a,b\in Y$,
\begin{equation}
\label{E6}
a\ge b \quad\Longleftrightarrow\quad \tau(a)\le \tau(b).
\end{equation}
For $Y^- = \{-1,-2,-3,\cdots\}$ one may take simply
\begin{equation*}
\tau(a)=-a.
\end{equation*}
We extend $\tau$ coordinate-wise to infinite utility streams $x$ by writing
\begin{equation*}
\tau(x)=(\tau(x_1),\tau(x_2),\cdots).
\end{equation*}

In the next proposition, we use the order-reversing transformation above to obtain the corresponding social welfare function when $Y$ is of order type $\omega^{\ast}$.

\begin{proposition}
\label{P2}

\emph{Let $Y\subset\mathbb{R}$ have order type $\omega^{\ast}$ and let $X=Y^{\mathbb{N}}$.
Then there exists a monotone social welfare function $W_Y:X\to\mathbb{R}$ satisfying finite anonymity, Almost Weak Pareto, and stationarity.
In particular, such a function exists for
\begin{equation*}
Y^- = \{-1,-2,-3,\cdots\}.
\end{equation*}}
\end{proposition}

\begin{proof}
Let $W_{\mathbb{N}}:\mathbb{N}^{\mathbb{N}}\to\mathbb{R}$ be the monotone social welfare function constructed in Proposition \ref{P1}.
Let $\tau:Y\to\mathbb{N}$ be an order-reversing bijection satisfying \eqref{E6}, and extend $\tau$ coordinatewise to $X=Y^{\mathbb{N}}$.
Define
\begin{equation}
\label{E7}
W_Y(x)=-W_{\mathbb{N}}(\tau(x)) \quad\text{for every }x\in X.
\end{equation}
This definition is well-defined because $\tau(x)\in\mathbb{N}^{\mathbb{N}}$ for every $x\in Y^{\mathbb{N}}$.

We first verify finite anonymity.
Let $\pi$ be a finite permutation of $\mathbb{N}$.
Since $\tau$ is applied coordinate-wise,
\begin{equation*}
\tau(x_{\pi})=\tau(x)_{\pi}.
\end{equation*}
By finite anonymity of $W_{\mathbb{N}}$,
\begin{equation*}
W_{\mathbb{N}}(\tau(x_{\pi}))=W_{\mathbb{N}}(\tau(x)_{\pi})=W_{\mathbb{N}}(\tau(x)).
\end{equation*}
Therefore, by \eqref{E7},
\begin{equation*}
W_Y(x_{\pi})=-W_{\mathbb{N}}(\tau(x_{\pi}))=-W_{\mathbb{N}}(\tau(x))=W_Y(x).
\end{equation*}
Thus $W_Y$ satisfies finite anonymity.

We next verify stationarity.
Take $x,y\in X$ with $x_1=y_1$.
Then
\begin{equation*}
\tau(x)_1=\tau(y)_1.
\end{equation*}
Let
\begin{equation*}
x^{\prime}=(x_2,x_3,\cdots), \quad y^{\prime}=(y_2,y_3,\cdots).
\end{equation*}
Since $\tau$ is applied coordinate-wise,
\begin{equation*}
\tau(x^{\prime})=(\tau(x))^{\prime}, \quad
\tau(y^{\prime})=(\tau(y))^{\prime}.
\end{equation*}
By stationarity of $W_{\mathbb{N}}$, applied to $\tau(x)$ and $\tau(y)$, one has
\begin{equation*}
W_{\mathbb{N}}(\tau(x))\le W_{\mathbb{N}}(\tau(y))
\quad\Longleftrightarrow\quad W_{\mathbb{N}}(\tau(x^{\prime}))\le W_{\mathbb{N}}(\tau(y^{\prime})).
\end{equation*}
Using \eqref{E7}, multiplying by $-1$ reverses the inequalities, and therefore
\begin{equation*}
W_Y(x)\ge W_Y(y) \quad\Longleftrightarrow\quad W_Y(x^{\prime})\ge
W_Y(y^{\prime}).
\end{equation*}
Thus $W_Y$ satisfies stationarity.

We next verify monotonicity.
Suppose $x\ge y$ coordinatewise in $Y^{\mathbb{N}}$.
Since $\tau$ reverses the order, we have
\begin{equation*}
\tau(y)\ge \tau(x)
\end{equation*}
coordinatewise in $\mathbb{N}^{\mathbb{N}}$.
By monotonicity of $W_{\mathbb{N}}$,
\begin{equation*}
W_{\mathbb{N}}(\tau(y))\ge W_{\mathbb{N}}(\tau(x)).
\end{equation*}
Using \eqref{E7}, multiplying by $-1$ reverses the inequality, so
\begin{equation*}
W_Y(x)=-W_{\mathbb{N}}(\tau(x))\ge -W_{\mathbb{N}}(\tau(y))=W_Y(y).
\end{equation*}
Hence $W_Y$ is monotone.

Finally, we verify AWP.
Suppose $x,y\in X$ satisfy $x\ge y$ coordinatewise and that there exists a finite integer $\bar N\in\mathbb{N}$ such that
\begin{equation*}
x_t>y_t\quad\text{for all }t>\bar N.
\end{equation*}
Since $\tau$ reverses the order, we have
\begin{equation*}
\tau(y)\ge \tau(x)
\end{equation*}
coordinatewise, and also
\begin{equation*}
\tau(y)_t>\tau(x)_t \quad\text{for all }t>\bar N.
\end{equation*}
By AWP for $W_{\mathbb{N}}$ applied to the pair $(\tau(y),\tau(x))$,
\begin{equation*}
W_{\mathbb{N}}(\tau(y))>W_{\mathbb{N}}(\tau(x)).
\end{equation*}
Therefore, using \eqref{E7},
\begin{equation*}
W_Y(x)=-W_{\mathbb{N}}(\tau(x)) > -W_{\mathbb{N}}(\tau(y))=W_Y(y).
\end{equation*}
Thus $W_Y$ satisfies AWP.
\end{proof}
The result of Proposition \ref{P2} can be extended from order type $\omega^{\ast}$ to an arbitrary non-empty reverse well ordered subset of $\mathbb{R}$, ordered by the usual $<$ relation.

\begin{lemma}
\label{L6}
\emph{Let $Y\subseteq\mathbb{R}$ be an arbitrary non-empty reverse well ordered set with respect to the usual $<$ relation on $\mathbb{R}$, and let $X=Y^{\mathbb{N}}$.
Then there exists a monotone social welfare function $W:X\to\mathbb{R}$ satisfying Anonymity, Almost Weak Pareto, and stationarity.}
\end{lemma}

\begin{proof}
See Appendix \ref{app}.
\end{proof}

\subsection{SWF for $Y$ of order type $\protect\omega + \protect\omega^{\ast} $}

We now consider the case in which $Y$ has order type $\omega+\omega^{\ast}$.
Thus $Y=C\cup B$, where
\begin{equation*}
C= \{c_1, c_2, \cdots\}, \qquad c_1<c_2<\cdots,
\end{equation*}
is of order type $\omega$,
\begin{equation*}
B= \{b_1, b_2, \cdots\}, \qquad b_1>b_2>\cdots,
\end{equation*}
is of order type $\omega^{\ast}$, and $b_k>c_l$ for all $k,l\in\mathbb{N}$.\footnote{
For an example, consider $B= \left\{b_n = \frac{n+3}{2(n+1)}: n\in \mathbb{N}\right\}$ and $C= \left\{c_n = \frac{n-1}{2(n+1)}: n\in \mathbb{N}\right\}$.}
For $x\in X$ and $i\in\mathbb{N}$, define
\begin{equation}
\label{E8}
m_i(x)=\left| \{t\in\mathbb{N}:x_t=b_i\}\right|, \quad n_i(x)=\left| \{t\in\mathbb{N}:x_t=c_i\}\right|.
\end{equation}
Thus $m_i(x),n_i(x)\in\mathcal{A}$.
The number $m_i(x)$ counts the occurrences of the upper utility level $b_i$.
The number $n_i(x)$ counts the occurrences of the lower utility level $c_i$.

\begin{definition}
\label{D8}
For $x\in X$, using the multiplicities in \eqref{E8}, define $p(x)\in\mathcal{A}^{\mathbb{N}}$ as follows.
For every $i\in\mathbb{N}$,
\begin{equation*}
p_{2i-1}(x)=
\begin{cases}
m_i(x), & m_h(x)<\infty\text{ for every }h<i, \\
\infty, & m_h(x)=\infty\text{ for some }h<i,\end{cases}\end{equation*}
and
\begin{equation*}
p_{2i}(x)=
\begin{cases}
n_i(x), & n_h(x)<\infty\text{ for every }h<i, \\
\infty, & n_h(x)=\infty\text{ for some }h<i.\end{cases}\end{equation*}
Let
\begin{equation*}
\mathcal{P}=\{p(x):x\in X\}.
\end{equation*}
\end{definition}

Odd profile coordinates count upper utility levels.
On odd coordinates, larger multiplicity is better.
Even profile coordinates count lower utility levels.
On even coordinates, smaller multiplicity is better.

\begin{definition}[Lexicographic profile order]
\label{D9}
For $p,q\in\mathcal{A}^{\mathbb{N}}$, define $p\succeq q $ if either $p=q$, or, letting
\begin{equation*}
k=\min\{i\in\mathbb{N}:p_i\ne q_i\},
\end{equation*}
one has
\begin{equation*}
\begin{cases}
p_k>q_k, & k\text{ odd}, \\
p_k<q_k, & k\text{ even}.\end{cases}\end{equation*}
Write $p\succ q$ if $p\succeq q$ and $p\ne q$.
\end{definition}

In the next lemma, we show that the lexicographic profile order admits a countable order-separating family.

\begin{lemma}
\label{L7}
\emph{There exists a countable set $\mathcal{D}\subset\mathcal{A}^{\mathbb{N}}$ such that, for all $p,q\in\mathcal{P}$, if $p\succ q$, then
there exists $d\in\mathcal{D}$ such that
\begin{equation*}
p\succeq d\succ q.
\end{equation*}}
\end{lemma}

\begin{proof}
For every finite sequence $a=(a_1,\cdots,a_k)\in\mathcal{A}^k$, define $d(a)\in\mathcal{A}^{\mathbb{N}}$ as follows.
For $i\le k$, set
\begin{equation*}
d_i(a)=a_i.
\end{equation*}
For $i>k$, set
\begin{equation*}
d_i(a)=
\begin{cases}
0, & i\text{ odd and no odd }h\le k\text{ satisfies }a_h=\infty, \\
\infty, & i\text{ odd and some odd }h\le k\text{ satisfies }a_h=\infty, \\
\infty, & i\text{ even}.\end{cases}\end{equation*}
Let
\begin{equation*}
\mathcal{D}=\{d(a):a\in\mathcal{A}^k\text{ for some }k\in\mathbb{N}\}.
\end{equation*}
The set $\mathcal{A}$ is countable.
Therefore the set of finite sequences over $\mathcal{A}$ is countable.
Hence $\mathcal{D}$ is countable.

Take $p,q\in\mathcal{P}$ with $p\succ q$.
Let
\begin{equation*}
k=\min\{i\in\mathbb{N}:p_i\ne q_i\}.
\end{equation*}
Define
\begin{equation*}
d=d(p_1,\cdots,p_k).
\end{equation*}
Then $d\in\mathcal{D}$.
For every $i<k$, one has $d_i=p_i=q_i$.
At coordinate $k$, one has $d_k=p_k$.
Since $p\succ q$, the comparison between $d$ and $q$ is decided at coordinate $k$.
Therefore $d\succ q$.

It remains to show $p\succeq d$.
The sequences $p$ and $d$ agree through coordinate $k$.
After coordinate $k$, the sequence $d$ has the worst available tail relative to the parity of the coordinates.
On odd coordinates, where larger values are better, the worst value is $0$ unless an odd $\infty$ has already appeared.
If an odd $\infty$ has already appeared in the prefix, all later odd coordinates of any canonical profile are also $\infty$.
On even coordinates, where smaller values are better, the worst value is $\infty$.
Thus, if $p$ and $d$ differ after coordinate $k$, then at their first later difference $p$ is better than $d$.
Hence $p\succeq d$.
Therefore $p\succeq d\succ q$.
\end{proof}

Using the order-separating family established in Lemma \ref{L7}, in the next lemma we show that the profile order admits a real-valued representation.

\begin{lemma}
\label{L8}
\emph{There exists a function $V:\mathcal{P}\to\mathbb{R}$ such that, for all $p,q\in\mathcal{P}$,
\begin{equation*}
V(p)\ge V(q) \quad\Longleftrightarrow\quad p\succeq q.
\end{equation*}}
\end{lemma}

\begin{proof}
Let $\mathcal{D}=\{d^1,d^2,\cdots\}$ be an enumeration of the countable set from Lemma \ref{L7}.
For $p\in\mathcal{P}$, define
\begin{equation*}
V(p)=\sum_{s=1}^{\infty}2^{-s}\mathbf{1}_{\{p\succeq d^s\}}.
\end{equation*}
The sum is well-defined because each summand is between $0$ and $2^{-s}$.
Moreover, $0\le V(p)\le 1$.

Suppose $p\succeq q$.
If $q\succeq d^s$, then transitivity gives $p\succeq d^s$.
Thus
\begin{equation*}
\mathbf{1}_{\{p\succeq d^s\}}\ge \mathbf{1}_{\{q\succeq d^s\}}
\end{equation*}
for every $s\in\mathbb{N}$.
Therefore $V(p)\ge V(q)$.

Suppose $p\succ q$.
By Lemma \ref{L7}, there exists $s^{\ast}\in\mathbb{N}$ such that
\begin{equation*}
p\succeq d^{s^{\ast}}\succ q.
\end{equation*}
For this index, one has
\begin{equation*}
\mathbf{1}_{\{p\succeq d^{s^{\ast}}\}}=1 \quad\text{and}\quad \mathbf{1}_{\{q\succeq d^{s^{\ast}}\}}=0.
\end{equation*}
For every other index $s$, the inequality
\begin{equation*}
\mathbf{1}_{\{p\succeq d^s\}}\ge \mathbf{1}_{\{q\succeq d^s\}}
\end{equation*}
follows from $p\succeq q$ and transitivity.
Hence $V(p)>V(q)$.

Finally, suppose $V(p)\ge V(q)$.
If $p\not\succeq q$, completeness gives $q\succ p$.
The previous paragraph would imply $V(q)>V(p)$, a contradiction.
Therefore $p\succeq q$.
\end{proof}

The next theorem records the two one-sided cases together with the two-sided case.

\begin{theorem}
\label{T1}
Let $Y\subset\mathbb{R}$ have order type $\omega$, $\omega^{\ast}$, or $\omega+\omega^{\ast}$, and let $X=Y^{\mathbb{N}}$.
Then there exists a monotone social welfare function $W:X\to\mathbb{R}$ satisfying AN, AWP, and stationarity.
\end{theorem}

\begin{proof}
First suppose that $Y$ has order type $\omega$.
Let $W_{\mathbb{N}}:\mathbb{N}^{\mathbb{N}}\to\mathbb{R}$ be the social welfare function constructed in Proposition \ref{P1}.
Choose an order-preserving bijection
\begin{equation*}
\sigma:Y\to\mathbb{N}.
\end{equation*}
Extend $\sigma$ coordinate-wise to streams by writing
\begin{equation*}
\sigma(x)=(\sigma(x_1),\sigma(x_2),\cdots).
\end{equation*}
Define
\begin{equation*}
W_Y(x)=W_{\mathbb{N}}(\sigma(x)) \quad\text{for every }x\in Y^{\mathbb{N}}.
\end{equation*}
If $\pi$ is a finite permutation of $\mathbb{N}$, then
\begin{equation*}
\sigma(x_{\pi})=\sigma(x)_{\pi}.
\end{equation*}
Finite anonymity of $W_{\mathbb{N}}$ therefore implies finite anonymity of $W_Y$.
Since $\sigma$ preserves the order, $x\ge y$ implies $\sigma(x)\ge \sigma(y)$, and a strict tail improvement of $x$ over $y$ is carried into a strict tail improvement of $\sigma(x)$ over $\sigma(y)$.
Thus monotonicity and AWP of $W_Y$ follow from monotonicity and AWP of $W_{\mathbb{N}}$.
Finally, if $x_1=y_1$, then $\sigma(x)_1=\sigma(y)_1$ and
\begin{equation*}
\sigma(x^{\prime})=(\sigma(x))^{\prime}, \quad
\sigma(y^{\prime})=(\sigma(y))^{\prime}.
\end{equation*}
By stationarity of $W_{\mathbb{N}}$,
\begin{equation*}
W_Y(x)\ge W_Y(y) \quad\Longleftrightarrow\quad W_Y(x^{\prime})\ge
W_Y(y^{\prime}).
\end{equation*}
Thus $W_Y$ satisfies stationarity, and the conclusion holds in the case of order type $\omega$.

If $Y$ has order type $\omega^{\ast}$, the conclusion follows from Proposition \ref{P2}.

It remains to consider the case in which $Y$ has order type $\omega+\omega^{\ast}$.
Let $V:\mathcal{P}\to\mathbb{R}$ be the representing function from Lemma \ref{L8}.
Define
\begin{equation}
\label{E9}
W(x)=V(p(x)) \quad\text{for every }x\in X.
\end{equation}

\begin{itemize}
\item {\textbf{Anonymity (AN).}
Let $\pi$ be a finite permutation of $\mathbb{N}$.
For every $i\in\mathbb{N}$, a permutation preserves the cardinality of the level sets $\{t\in\mathbb{N}:x_t=b_i\}$ and $\{t\in\mathbb{N}:x_t=c_i\}$.
Thus
\begin{equation*}
m_i(x_{\pi})=m_i(x) \quad\text{and}\quad n_i(x_{\pi})=n_i(x).
\end{equation*}
Therefore $p(x_{\pi})=p(x)$.
Hence, by \eqref{E9}, $W(x_{\pi})=W(x)$.
So $W$ satisfies AN.
Observe that even in the case of arbitrary permutations, the multiplicity of every utility level is preserved.
Hence $W$ satisfies the Strong Anonymity axiom as well.}

\item {\textbf{Stationarity.}
Take $x,y\in X$ with $x_1=y_1$.
Let
\begin{equation*}
x^{\prime}=(x_2,x_3,\cdots), \quad y^{\prime}=(y_2,y_3,\cdots).
\end{equation*}
We show that
\begin{equation}
\label{E10}
p(x)\succeq p(y) \quad\Longleftrightarrow\quad p(x^{\prime})\succeq
p(y^{\prime}).
\end{equation}
Suppose first that $x_1=y_1=b_r$ for some $r\in\mathbb{N}$.
Deleting the common first coordinate changes only the upper multiplicity at level $b_r$.
For every $i\ne r$, one has
\begin{equation*}
m_i(x^{\prime})=m_i(x), \quad m_i(y^{\prime})=m_i(y),
\end{equation*}
and for every $i\in\mathbb{N}$, one has
\begin{equation*}
n_i(x^{\prime})=n_i(x), \quad n_i(y^{\prime})=n_i(y).
\end{equation*}
At level $b_r$, deleting the common first coordinate replaces each finite positive value $a$ by $a-1$ and leaves $\infty$ equal to $\infty$.
This operation preserves equality, strict inequality, and the property of being infinite on $\{1,2,3,\cdots\}\cup\{\infty\}$.
Therefore the odd profile coordinates are compared the same way before and after deletion.
The even profile coordinates are unchanged.
Hence \eqref{E10} holds.
}

Suppose next that $x_1=y_1=c_r$ for some $r\in\mathbb{N}$.
Deleting the common first coordinate changes only the lower multiplicity at level $c_r$.
For every $i\ne r$, one has
\begin{equation*}
n_i(x^{\prime})=n_i(x), \quad n_i(y^{\prime})=n_i(y),
\end{equation*}
and for every $i\in\mathbb{N}$, one has
\begin{equation*}
m_i(x^{\prime})=m_i(x), \quad m_i(y^{\prime})=m_i(y).
\end{equation*}
At level $c_r$, deleting the common first coordinate replaces each finite positive value $a$ by $a-1$ and leaves $\infty$ equal to $\infty$.
This operation preserves equality, strict inequality, and the property of being infinite on $\{1,2,3,\cdots\}\cup\{\infty\}$.
Therefore the even profile coordinates are compared the same way before and after deletion, remembering that smaller values are better on even coordinates.
The odd profile coordinates are unchanged.
Hence again \eqref{E10} holds.
By Lemma \ref{L8} and \eqref{E9}, thus,
\begin{equation*}
W(x)\ge W(y) \quad\Longleftrightarrow\quad W(x^{\prime})\ge W(y^{\prime}).
\end{equation*}
Therefore $W$ satisfies stationarity.

\item {\textbf{Monotonicity.}
Suppose $x\ge y$.
Then, for every $r\in\mathbb{N}$,
\begin{equation*}
\{t\in\mathbb{N}:y_t\ge b_r\}\subseteq\{t\in\mathbb{N}:x_t\ge b_r\},\;\text{and}\; \{t\in\mathbb{N}:x_t\le c_r\}\subseteq\{t\in\mathbb{N}:y_t\le c_r\}.
\end{equation*}
Hence
\begin{equation}
\label{E11}
\left| \{t\in\mathbb{N}:x_t\ge b_r\}\right| \ge \left| \{t\in\mathbb{N}:y_t\ge b_r\}\right|,\;\text{and}\; \left| \{t\in\mathbb{N}:x_t\le
c_r\}\right| \le \left| \{t\in\mathbb{N}:y_t\le c_r\}\right|
\end{equation}
for every $r\in\mathbb{N}$.
If $p(x)=p(y)$, then $W(x)=W(y)$.
Suppose $p(x)\ne p(y)$.
Let
\begin{equation*}
k=\min\{i\in\mathbb{N}:p_i(x)\ne p_i(y)\}.
\end{equation*}
If $k=2r-1$ is odd, then all earlier upper multiplicities agree and are finite.
Thus
\begin{equation*}
\left| \{t\in\mathbb{N}:x_t\ge b_{r-1}\}\right| = \left| \{t\in\mathbb{N}:y_t\ge b_{r-1}\}\right|
\end{equation*}
when $r>1$.
When $r=1$, the corresponding equality is the equality $0=0$ for the empty upper cutoff above $b_1$.
Together with the first inequality in \eqref{E11} and $p_k(x)\ne p_k(y)$, one obtains $m_r(x)>m_r(y)$.
Therefore $p_k(x)>p_k(y) $, so $p(x)\succ p(y)$.
If $k=2r$ is even, then all earlier lower multiplicities agree and are finite.
Thus
\begin{equation*}
\left| \{t\in\mathbb{N}:x_t\le c_{r-1}\}\right| = \left| \{t\in\mathbb{N}:y_t\le c_{r-1}\}\right|
\end{equation*}
when $r>1$.
When $r=1$, the corresponding equality is the equality $0=0$ for the empty lower cutoff below $c_1$.
Together with the second inequality in \eqref{E11} and $p_k(x)\ne p_k(y)$, one obtains $n_r(x)<n_r(y)$.
Therefore $p_k(x)<p_k(y) $, so $p(x)\succ p(y)$.
In all cases $p(x)\succeq p(y)$.
By Lemma \ref{L8} and \eqref{E9}, $W(x)\ge W(y)$.
Thus $W$ is monotone.}

\item {\textbf{Almost Weak Pareto (AWP).}
Suppose $x\ge y$ and there exists $\bar N\in\mathbb{N}$ such that $x_t>y_t$ for every $t>\bar N$.
By the monotonicity argument based on \eqref{E11}, $p(x)\succeq p(y)$.
It is enough to prove $p(x)\ne p(y)$.
We will show that at least one cumulative cutoff inequality is strict. }

First suppose that some tail value of $x$ lies in $C$.
Then the nonempty set
\begin{equation*}
\{x_t:t>\bar N,\ x_t\in C\}
\end{equation*}
has a least element.
Let
\begin{equation*}
c_{r^{\ast}}=\min\{x_t:t>\bar N,\ x_t\in C\}.
\end{equation*}
Choose $t^{\ast}>\bar N$ such that $x_{t^{\ast}}=c_{r^{\ast}}$.
Since $x_{t^{\ast}}>y_{t^{\ast}}$, one has $y_{t^{\ast}}<c_{r^{\ast}}$.
The case $r^{\ast}=1$ is impossible because $c_1$ is the minimum element of $Y$.
Thus $r^{\ast}>1$, and $y_{t^{\ast}}\le c_{r^{\ast}-1}$.
By minimality of $c_{r^{\ast}}$, no tail value of $x$ is below $c_{r^{\ast}}$.
Therefore $\{t\in\mathbb{N}:x_t\le c_{r^{\ast}-1}\}$ is finite.
Moreover, $t^{\ast}$ belongs to $\{t\in\mathbb{N}:y_t\le c_{r^{\ast}-1}\}$ but not to $\{t\in\mathbb{N}:x_t\le c_{r^{\ast}-1}\}$.
Hence
\begin{equation*}
\left| \{t\in\mathbb{N}:x_t\le c_{r^{\ast}-1}\}\right| < \left| \{t\in\mathbb{N}:y_t\le c_{r^{\ast}-1}\}\right|.
\end{equation*}
So a lower cumulative cutoff inequality is strict.

Now suppose that every tail value of $x$ lies in $B$.
If some tail value of $y$ lies in $B$, then the nonempty set
\begin{equation*}
\{y_t:t>\bar N,\ y_t\in B\}
\end{equation*}
has a maximum element.
Let
\begin{equation*}
b_{r^{\ast}}=\max\{y_t:t>\bar N,\ y_t\in B\}.
\end{equation*}
Choose $t^{\ast}>\bar N$ such that $y_{t^{\ast}}=b_{r^{\ast}}$.
Since $x_{t^{\ast}}>y_{t^{\ast}}$, one has $x_{t^{\ast}}>b_{r^{\ast}}$.
The case $r^{\ast}=1$ is impossible because $b_1$ is the maximum element of $Y$.
Thus $r^{\ast}>1$, and $x_{t^{\ast}}\ge b_{r^{\ast}-1}$.
By maximality of $b_{r^{\ast}}$, no tail value of $y$ is above $b_{r^{\ast}}$.
Therefore $\{t\in\mathbb{N}:y_t\ge b_{r^{\ast}-1}\}$ is finite.
Moreover, $t^{\ast}$ belongs to $\{t\in\mathbb{N}:x_t\ge b_{r^{\ast}-1}\}$ but not to $\{t\in\mathbb{N}:y_t\ge b_{r^{\ast}-1}\}$.
Hence
\begin{equation*}
\left| \{t\in\mathbb{N}:x_t\ge b_{r^{\ast}-1}\}\right| > \left| \{t\in\mathbb{N}:y_t\ge b_{r^{\ast}-1}\}\right|.
\end{equation*}
So an upper cumulative cutoff inequality is strict.

Finally suppose that every tail value of $x$ lies in $B$ and every tail value of $y$ lies in $C$.
Choose any $t^{\ast}>\bar N$.
Write $x_{t^{\ast}}=b_r$.
Then $t^{\ast}$ belongs to $\{t\in\mathbb{N}:x_t\ge b_r\}$.
Since every tail value of $y$ lies in $C$, the set $\{t\in\mathbb{N}:y_t\ge b_r\}$ is finite.
Since $x\ge y$, one has $\{t\in\mathbb{N}:y_t\ge b_r\}\subseteq\{t\in\mathbb{N}:x_t\ge b_r\}$.
The element $t^{\ast}$ belongs to the latter set but not to the former set.
Hence
\begin{equation*}
\left| \{t\in\mathbb{N}:x_t\ge b_r\}\right| > \left| \{t\in\mathbb{N}:y_t\ge
b_r\}\right|.
\end{equation*}
So an upper cumulative cutoff inequality is strict.

In each case, at least one cumulative cutoff inequality is strict.
By the monotonicity argument based on \eqref{E11}, a strict cumulative cutoff inequality implies that the corresponding profiles cannot be equal.
Hence $p(x)\ne p(y)$.
Since $p(x)\succeq p(y)$ and $p(x)\ne p(y)$, one has $p(x)\succ p(y)$.
By Lemma \ref{L8} and \eqref{E9}, $W(x)>W(y)$.
Thus $W$ satisfies AWP.
\end{itemize}

\end{proof}
We generalize the case of $Y$ having order type $\omega+\omega^{\ast}$ as follows.
Let $Y\subset\mathbb{R}$, ordered by the usual $<$ relation, be such that every non-empty subset of $Y$ has either a first or a last element.
By Lemma \ref{L2}, $Y$ can be represented as the ordered sum of two components, of which the first is a well ordered (or empty) set and the second is a reverse well ordered (or empty) set.
The two components need not have the same order type.
The following lemma gives a direct extension to this case by combining lower and upper cumulative cutoff counts for the two components.

\begin{lemma}
\label{L9}
\emph{Let $Y\subset\R$ be non-empty, ordered by the usual $<$ relation, and such that every non-empty subset of $Y$ has either a first or a last element, and let $X=Y^{\mathbb{N}}$.
There exists a monotone social welfare function $W:X\to\mathbb{R}$ satisfying Almost Weak Pareto, Anonymity, and stationarity.}
\end{lemma}

\begin{proof}
See Appendix \ref{app}.
\end{proof}

\begin{remark}
\label{R1}
\emph{Unlike the SWF in Theorem \ref{T1} that satisfies stationarity, the SWF in \citet[Proposition 1]{dubey2011}, described below, does not satisfy stationarity as the following example shows.
\begin{equation}
\label{E12}
W(x)= \alpha \inf\{x_t: t\in\mathbb{N}\} +(1-\alpha) \sup\{x_t: t\in\mathbb{N}\}, \;\emph{where}\; \alpha\in (0,1).
\end{equation}
For the function in \eqref{E12}, let $x=\left(\frac{1}{2}, \frac{2}{3}, \frac{3}{4}, \cdots\right)$, and $y=\left(\frac{1}{2}, \frac{3}{4}, \frac{4}{5}, \cdots\right)$.
Then
\begin{equation*}
\inf\{x_t: t\in\mathbb{N}\} =\frac{1}{2} = \inf\{y_t: t\in\mathbb{N}\}, \;\emph{and}\; \sup\{x_t: t\in\mathbb{N}\} =1 = \sup\{y_t: t\in\mathbb{N}\}.
\end{equation*}
Hence,
\begin{equation*}
W(x)=\alpha \left(\frac{1}{2}\right) + (1-\alpha) 1=1-\frac{\alpha}{2}= W(y).
\end{equation*}
Thus $W(x)\ge W(y)$.
However, after deleting the common first coordinate, we obtain $x^{\prime} = \left(\frac{2}{3}, \frac{3}{4}, \frac{4}{5}, \cdots\right)$, and $y^{\prime}=\left(\frac{3}{4}, \frac{4}{5}, \frac{5}{6}, \cdots\right)$ and
\begin{equation*}
W(x^{\prime})=\alpha \left(\frac{2}{3}\right) + (1-\alpha) 1=1-\frac{\alpha}{3}< 1-\frac{\alpha}{4}=\alpha \left(\frac{3}{4}\right) + (1-\alpha) 1 =
W(y^{\prime}).
\end{equation*}
Hence $W(x^{\prime})\ge W(y^{\prime})$ is false.
Therefore, $W(x)\ge W(y)$ $\centernot{\Leftrightarrow}\quad W(x^{\prime})\ge W(y^{\prime})$, even though $x_1=y_1$.
Hence $W$ violates stationarity.}
\end{remark}

\begin{remark}
\label{R2}
\emph{For a domain $Y\subset\mathbb{R}$ containing a subset that is order-isomorphic to the set of integers, \citet{dubey2011} proves that there does not exist a SWF satisfying AN and WP.
Since AWP is a stronger Pareto condition than WP, the above result implies the impossibility of real-valued representation of SWOs satisfying AWP and AN.
Thus we are able to conclude that the domain restrictions on $Y$ for the existence of a SWF satisfying AN and AWP are identical to those for the existence of a SWF satisfying AN and WP.
We state it in the following corollary.}
\end{remark}

\begin{corollary}
\label{C1}
\emph{Let $Y\subset \mathbb{R}$ be non-empty, and let $X=Y^{\mathbb{N}}$.
There exists a social welfare function $W:X\rightarrow \mathbb{R}$ satisfying Almost Weak Pareto and Anonymity if and only if $Y$ does not contain a subset that is order-isomorphic to $\mathbb{Z}$.}
\end{corollary}

\begin{proof}
Suppose first that $Y$ does not contain a subset that is order-isomorphic to $\mathbb{Z}$.
Then, by Lemma \ref{L1}, every non-empty subset of $Y$ has either a first element or a last element.
It follows from Lemma \ref{L9} that there exists a social welfare function on $X=Y^{\mathbb{N}}$ satisfying Almost Weak Pareto and Anonymity.

Conversely, suppose that there exists a social welfare function $W:X\rightarrow \mathbb{R}$ satisfying Almost Weak Pareto and Anonymity and assume, toward obtaining a contradiction, that $Y$ contains a subset that is order-isomorphic to $\mathbb{Z}$.
Since Almost Weak Pareto implies Weak Pareto, such a social welfare function would also satisfy Weak Pareto and Anonymity, which contradicts the characterization in \citet{dubey2011}, and we are done.
\end{proof}

In what follows, we use ``non-constructive'' in the sense of \citet{dubey2011a}: the existence of the social welfare order entails the existence of a non-Ramsey set, while a constructively defined order admits an explicit description without invoking such a device.

\begin{remark}
\label{R3}
\emph{For the domain $Y=[0, 1]$, \citet{svensson1980} has established that there exists a social welfare order satisfying anonymity and strong Pareto axioms.
Since the strong Pareto axiom\footnote{
A social welfare order $\succeq$ satisfies Strong Pareto if, for all $x,y\in X$, $x>y$ implies $x\succ y$.} is stronger than AWP, we could also infer that the same social welfare order would satisfy anonymity as well as AWP.
However, \citet{lauwers2010} showed that the existence of any social welfare order satisfying strong Pareto and anonymity entails the existence of a non-Ramsey set, which is known to be a non-constructive object.
Therefore, even though the social welfare order satisfying strong Pareto and anonymity axioms exists, it does not admit explicit description.
This result was extended to the social welfare order satisfying anonymity and weak Pareto when the domain $Y$ contains a subset which is order-isomorphic to the set of integers in \citet[Theorems 1 and 2]{dubey2011a}, which is precisely the domain on which there does not exist a social welfare function satisfying anonymity and Weak Pareto.
Corollary \ref{C1} implies the non-existence of a social welfare function satisfying anonymity and AWP on the identical domain.
Since AWP is a stronger Pareto condition than Weak Pareto, we are able to conclude that the social welfare order satisfying anonymity and AWP is a non-constructive object in the precise sense stated in the following corollary.}
\end{remark}

\begin{corollary}
\label{C2}
\emph{Let $Y\subset \mathbb{R}$ be non-empty and suppose that $Y$ contains a subset that is order-isomorphic to $\mathbb{Z}$, and let $X=Y^{\mathbb{N}}$.
Then the existence of any Almost Weak Paretian social welfare order $\succcurlyeq$ on $X$ satisfying Anonymity entails the existence of a non-Ramsey set; hence every such order is a non-constructive object.}
\end{corollary}

\begin{proof}
Let $\succcurlyeq$ be a social welfare order on $X=Y^{\mathbb{N}}$ satisfying Almost Weak Pareto and Anonymity.
Since Almost Weak Pareto implies Weak Pareto, $\succcurlyeq$ also satisfies Weak Pareto and Anonymity.
Because $Y$ contains a subset that is order-isomorphic to $\mathbb Z$, \citet[Theorem 1]{dubey2011a} implies that the existence of such a social welfare order entails the existence of a non-Ramsey set.
Hence $\succcurlyeq$ is non-constructive in the stated sense.
\end{proof}
Combining Corollaries \ref{C1} and \ref{C2} leads to the following characterization result.

\begin{theorem}
\label{T2}
\emph{Let $Y\subset \mathbb{R}$ be non-empty, and let $X=Y^{\mathbb{N}}$.
There does not exist any SWF $W:X\rightarrow \mathbb{R}$ satisfying AN and AWP if and only if every SWO on $X$ satisfying AN and AWP is non-constructive.}
\end{theorem}

\begin{proof}
Necessity: Assume, first, that there does not exist any SWF $W:X\rightarrow \mathbb{R}$ satisfying Anonymity and Almost Weak Pareto.
Then, by Corollary \ref{C1}, $Y$ contains a subset that is order-isomorphic to $\mathbb{Z}$.
Therefore, it readily follows from Corollary \ref{C2} that every SWO on $X$ which satisfies AN and AWP is a non-constructive object.

Sufficiency: Suppose that every SWO on $X$ satisfying AN and AWP is non-constructive.
If $Y$ did not contain a subset that is order-isomorphic to $\mathbb{Z}$, then, by Lemma \ref{L1}, every non-empty subset of $Y$ would have either a first element or a last element.
Lemma \ref{L9} would then provide an explicitly defined SWF $W:X\rightarrow \mathbb{R}$ satisfying AN and AWP, constructed from countable cutoff profiles and a convergent series without invoking a non-Ramsey set.
The SWO defined on $X$ by $x\succeq y$ if and only if $W(x)\geq W(y)$ would therefore be a constructive SWO satisfying AN and AWP, which is a contradiction.
Hence $Y$ contains a subset that is order-isomorphic to $\mathbb{Z}$.
It follows from Corollary \ref{C1} that there does not exist any SWF $W:X\rightarrow \mathbb{R}$ satisfying AN and AWP, and the proof is finished.
\end{proof}

\section{Conclusion}

\label{sec:4}
In this paper we have studied the existence and construction of social welfare functions on infinite utility streams satisfying anonymity and Almost Weak Pareto.
We have characterized the restriction on the feasible set of one-period utilities $Y\subset\mathbb{R}$ under which such a social welfare function exists.
The necessary and sufficient condition is that $Y$ does not contain a subset that is order-isomorphic to the set of integers.
This condition is identical to the domain restriction obtained in \citet{dubey2011} for the existence of a SWF satisfying anonymity and Weak Pareto.
Thus, strengthening WP to AWP does not change the boundary between possibility and impossibility with respect to the domain $Y$.

The possibility part of our characterization is constructive.
The SWF we construct is monotone, invariant under arbitrary permutations, and satisfies stationarity in addition to AWP.
On the complementary class of domains, not only does a real-valued representation satisfying AN and AWP fail to exist, but every SWO satisfying these axioms is non-constructive in the sense that its existence entails the existence of a non-Ramsey set.
Thus, either there exists an explicitly defined social welfare function satisfying AN and AWP, or every social welfare order satisfying these axioms is non-constructive.

Viewed in relation to the existing literature, AWP lies between ADOP and WP but retains the same domain restriction as WP.
In contrast, a SWF satisfying anonymity and ADOP exists only when $Y$ is finite.
The tables below summarize these results and place them within the existing possibility and impossibility results.
It is also easy to verify that the SWF in \citet{petri2019} satisfies a stationarity axiom similar to the SWF in Theorem \ref{T1}.

\begin{table}[!ht]
\centering
\begin{threeparttable}
\caption{Non-existence of Ethical social welfare functions}
\label{tab1}

{\small
\setlength{\tabcolsep}{3pt}
\begin{tabular*}{\textwidth}{@{\extracolsep{\fill}} l c l @{}}
\toprule
Efficiency axiom & $Y$  &Source \\
\midrule
Pareto & $|Y|\geq 2$ &  \citet{basu2003} \\
Infinite Pareto & $|Y|\geq 2$ &  \citet{crespo2009} \\
Upper asymptotic Pareto & $|Y|\geq 2$ &  \citet{dubey2020} \\
Lower asymptotic Pareto & $|Y|=\infty$ &  \citet{petri2019} \\
Asymptotic density-one Pareto & $|Y|=\infty$ &  \citet{dubey2021} \\
Almost Weak Pareto & $Y^{\ast\ast}$ & Corollary \ref{C1} \\
Weak Pareto & $Y^{\ast\ast}$ &  \citet{dubey2011} \\
Uniform Pareto & $Y$ bounded &  \citet{sakai2012} \\
\bottomrule
\end{tabular*}}

\begin{tablenotes}[flushleft]
\footnotesize
\item \emph{Note:} $Y^{\ast}\subset\R$ does not contain a subset that is order-isomorphic to the set of integers; and $Y^{\ast\ast}\subset\R$ contains a subset that is order-isomorphic to the set of integers.
\end{tablenotes}

\end{threeparttable}
\end{table}

\begin{table}[!ht]
\centering
\begin{threeparttable}
\caption{Ethical social welfare functions exist}
\label{tab2}

{\small
\setlength{\tabcolsep}{3pt}
\begin{tabular*}{\textwidth}{@{\extracolsep{\fill}} l c l @{}}
\toprule
Efficiency axiom & $Y$  & Source \\
\midrule
Lower asymptotic Pareto & $|Y|<\infty$  & \citet{petri2019} \\
Asymptotic density-one Pareto & $|Y|<\infty$  & \citet{petri2019} \\
Almost Weak Pareto & $Y^{\ast}$  & Corollary \ref{C1} \\
Weak Pareto & $Y^{\ast}$  & \citet{dubey2011} \\
\bottomrule
\end{tabular*}
}

\begin{tablenotes}[flushleft]
\footnotesize
\item \emph{Note:} $Y^{\ast}\subset\R$ does not contain a subset that is order-isomorphic to the set of integers.
\end{tablenotes}

\end{threeparttable}
\end{table}

\clearpage
\appendix
\section{Proofs of the lemma \ref{L5}, \ref{L6}, and \ref{L9}}
\label{app}

Throughout these proofs, we first move the relevant subset of the real line into the open unit interval by means of the strictly increasing bijection
\begin{equation}
\label{E13}
\phi:\mathbb R\longrightarrow(0,1),
\qquad
\phi(a)=\frac12+\frac12\frac{a}{1+|a|}.
\end{equation}
For \(a<0\), one has \(\phi(a)=1/[2(1-a)]\), while for \(a\ge 0\), one has \(\phi(a)=1-1/[2(1+a)]\).  
Hence \(\phi\) is strictly increasing, with limits \(0\) and \(1\) at \(-\infty\) and \(+\infty\), respectively, and therefore is a bijection from \(\mathbb R\) onto \((0,1)\).  
For an infinite utility stream \(x=(x_1,x_2,\cdots)\), write
\begin{equation*}
\phi(x)=(\phi(x_1),\phi(x_2),\cdots).
\end{equation*}
We also fix, for the remainder of the proof, a one-to-one enumeration (without repetitions) of the rationals in the open unit interval,
\begin{equation}
\label{E14}
\mathbb Q\cap(0,1)=\{q_1,q_2,q_3,\cdots\}.
\end{equation}

\subsection{Proof of Lemma \ref{L5}}
\begin{proof}[Proof of Lemma \ref{L5}]
Let
\begin{equation*}
\widehat Y=\phi(Y)\subset(0,1).
\end{equation*}
Since \(\phi\) is strictly increasing, \(\widehat Y\) is well ordered under the usual $<$ order.
We first show, using only the fixed enumeration in \eqref{E14}, that \(\widehat Y\) is countable.

For every \(z\in \widehat{Y}\) that is not the last element of \(\widehat{Y}\), define
\[
z^+ := \min\{y\in \widehat{Y} : z<y\}.
\]
This is well-defined because the set \(\{y\in \widehat{Y} : z<y\}\) is nonempty whenever \(z\) is not the last element of \(\widehat{Y}\), and every nonempty subset of the well-ordered set \(\widehat{Y}\) has a first element.
Because \((z,z^+)\subset(0,1)\) is a non-empty open interval, it contains a rational number from \(\mathbb Q\cap(0,1)\).
Define
\begin{equation*}
\iota(z)=\min\{j\in\mathbb N:q_j\in(z,z^+)\},
\end{equation*}
and, for every \(m\in\mathbb N\), set
\begin{equation}
\label{E15}
G_m=\{z\in\widehat Y:z\text{ is not the last element and }\iota(z)=m\}.
\end{equation}
The open intervals \((z,z^+)\), as \(z\) ranges over the elements of \(\widehat Y\) that are not the last element, are pairwise disjoint.
Consequently, a fixed rational \(q_m\) can be selected for at most one such interval, and hence every \(G_m\) has at most one element.\footnote{For a small illustration, suppose \(\widehat Y=\{1/4,1/2,3/4\}\) and one fixed enumeration of \(\mathbb Q\cap(0,1)\) begins \(q_1=1/3\), \(q_2=2/3\), \(\cdots\).  The interval \((1/4,1/2)\) first meets the enumeration at \(q_1\), so \(G_1=\{1/4\}\); the interval \((1/2,3/4)\) first meets it at \(q_2\), so \(G_2=\{1/2\}\).  The point \(3/4\) is the last element and is added separately.}
Thus every element of \(\widehat Y\) that is not the last element belongs to exactly one of the sets \(G_m\), and there is at most one last element.  Therefore
\begin{equation*}
\widehat Y=\bigcup_{m=1}^{\infty}G_m
\end{equation*}
if \(\widehat Y\) has no last element, while otherwise \(\widehat Y\) is this union together with its single last element.  Hence \(\widehat Y\) is countable.

For every \(z\in\widehat Y\) and every \(u\in\widehat Y^{\mathbb N}\), define the cutoff count
\begin{equation*}
C_z(u)=\left|\{t\in\mathbb N:u_t<z\}\right|\in\mathcal A.
\end{equation*}
Since \(\widehat Y\) is countable and non-empty, fix a sequence \((z_1,z_2,\cdots)\) in which every element of \(\widehat Y\) occurs at least once; if \(\widehat Y\) is finite, repeat its elements periodically.
For \(u\in\widehat Y^{\mathbb N}\), define the cutoff-count profile
\begin{equation*}
c(u)=\big(C_{z_1}(u),C_{z_2}(u),\cdots\big)\in\mathcal A^{\mathbb N}.
\end{equation*}
On \(\mathcal A^{\mathbb N}\), define the lexicographic relation \(\succeq_{\widehat Y}\), with a smaller entry at the first difference being better.
Thus, for \(p\ne q\), if
\begin{equation*}
k=\min\{i\in\mathbb N:p_i\ne q_i\},
\end{equation*}
then
\begin{equation*}
p\succ_{\widehat Y}q
\quad\Longleftrightarrow\quad
p_k<q_k.
\end{equation*}
This relation is complete and transitive.

We next establish a real-valued representation of this profile order.
For every finite sequence \(a=(a_1,\cdots,a_k)\in\mathcal A^k\), define
\begin{equation*}
d(a)=(a_1,\cdots,a_k,\infty,\infty,\cdots),
\end{equation*}
and let
\begin{equation*}
\mathcal D_{\widehat Y}
=\{d(a):a\in\mathcal A^k\text{ for some }k\in\mathbb N\}.
\end{equation*}
Since \(\mathcal A\) is countable, \(\mathcal D_{\widehat Y}\) is countable.
Moreover, \(\mathcal D_{\widehat Y}\) is order-separating.
Indeed, if \(p\succ_{\widehat Y}q\) and \(k\) is their first coordinate of difference, take
\begin{equation*}
d=d(p_1,\cdots,p_k).
\end{equation*}
Then \(d\succ_{\widehat Y}q\), because \(d_k=p_k<q_k\) and all earlier coordinates agree.
Also \(p\succeq_{\widehat Y}d\): if \(p\ne d\), then at their first difference after coordinate \(k\), the coordinate of \(d\) is \(\infty\), which is the worst possible cutoff count.
Hence
\begin{equation*}
p\succeq_{\widehat Y}d\succ_{\widehat Y}q.
\end{equation*}

Enumerate \(\mathcal D_{\widehat Y}=\{d^1,d^2,\cdots\}\) and define
\begin{equation*}
V_{\widehat Y}(p)=\sum_{s=1}^{\infty}2^{-s}
\mathbf 1_{\{p\succeq_{\widehat Y}d^s\}},
\qquad p\in\mathcal A^{\mathbb N}.
\end{equation*}
As in Lemma \ref{L4}, order separation implies that, for all \(p,q\in\mathcal A^{\mathbb N}\),
\begin{equation*}
V_{\widehat Y}(p)\ge V_{\widehat Y}(q)
\quad\Longleftrightarrow\quad
p\succeq_{\widehat Y}q.
\end{equation*}
Define
\begin{equation*}
\widehat W(u)=V_{\widehat Y}(c(u)),
\qquad u\in\widehat Y^{\mathbb N}.
\end{equation*}

An arbitrary permutation of \(\mathbb N\) preserves every cutoff count \(C_z(u)\).
Hence \(\widehat W\) satisfies strong anonymity, and therefore AN.

For stationarity, let \(u_1=v_1=a\), and write
\begin{equation*}
u'=(u_2,u_3,\cdots),
\qquad
v'=(v_2,v_3,\cdots).
\end{equation*}
For a cutoff \(z\le a\), deleting the common first coordinate leaves both cutoff counts unchanged.
For a cutoff \(z>a\), the common first coordinate is counted on both sides; hence both finite cutoff counts are positive, and deleting that coordinate subtracts one from each of them, while leaving \(\infty\) unchanged.
The map that sends a positive finite integer \(n\) to \(n-1\) and sends \(\infty\) to \(\infty\) preserves equality and strict order.  Consequently, equality and strict inequality at every cutoff coordinate are preserved.
Thus
\begin{equation*}
c(u)\succeq_{\widehat Y}c(v)
\quad\Longleftrightarrow\quad
c(u')\succeq_{\widehat Y}c(v'),
\end{equation*}
and stationarity follows.

For monotonicity, suppose \(u\ge v\).
Then, for every \(z\in\widehat Y\),
\begin{equation*}
\{t:u_t<z\}\subseteq\{t:v_t<z\},
\end{equation*}
so
\begin{equation*}
C_z(u)\le C_z(v).
\end{equation*}
If the two cutoff-count profiles are equal, then \(\widehat W(u)=\widehat W(v)\).
Otherwise, at their first difference the cutoff count of \(u\) is smaller, so
\begin{equation*}
c(u)\succ_{\widehat Y}c(v)
\end{equation*}
and \(\widehat W(u)>\widehat W(v)\).
Hence \(\widehat W\) is monotone.

Finally, suppose \(u\ge v\) and there exists \(\bar N\in\mathbb N\) such that
\begin{equation*}
u_t>v_t\qquad\text{for every }t>\bar N.
\end{equation*}
Since \(\widehat Y\) is well ordered, the non-empty set
\begin{equation*}
\{u_t:t>\bar N\}
\end{equation*}
has a first element, say \(z^{\ast}\).
Choose \(t^{\ast}>\bar N\) with \(u_{t^{\ast}}=z^{\ast}\).
Then \(v_{t^{\ast}}<z^{\ast}\), while \(u_t\ge z^{\ast}\) for every \(t>\bar N\).
Hence \(C_{z^{\ast}}(u)\) is finite, and
\begin{equation*}
C_{z^{\ast}}(u)<C_{z^{\ast}}(v).
\end{equation*}
Together with \(C_z(u)\le C_z(v)\) for every \(z\in\widehat Y\), this implies that the first coordinate at which \(c(u)\) and \(c(v)\) differ is favorable to \(u\).
Therefore
\begin{equation*}
c(u)\succ_{\widehat Y}c(v)
\end{equation*}
and \(\widehat W(u)>\widehat W(v)\).
Thus \(\widehat W\) satisfies AWP.

We now transfer the construction back to the original set \(Y\).
For \(x\in Y^{\mathbb N}\), define
\begin{equation*}
W(x)=\widehat W(\phi(x)).
\end{equation*}
Because \(\phi\) is strictly increasing and is applied coordinate-wise, it commutes with every permutation of coordinates, preserves equality of the first coordinate, preserves coordinatewise weak inequality, and preserves strict tail improvement.
Consequently, the anonymity, stationarity, monotonicity, and AWP properties of \(\widehat W\) transfer directly to \(W\).
This proves the lemma.
\end{proof}

\subsection{Proof of Lemma \ref{L6}}
\begin{proof}[Proof of Lemma \ref{L6}]
First map \(Y\) into the open unit interval by \(\phi\), and write
\begin{equation*}
\widehat Y=\phi(Y)\subset(0,1).
\end{equation*}
Since \(Y\) is reverse well ordered and \(\phi\) is strictly increasing, \(\widehat Y\) is reverse well ordered.
Now reflect \(\widehat Y\) inside \((0,1)\):
\begin{equation*}
Y^{\dagger}=1-\widehat Y
=\{1-u:u\in\widehat Y\}\subset(0,1).
\end{equation*}
The set \(Y^{\dagger}\) is well ordered.
Define
\begin{equation*}
\tau:Y\longrightarrow Y^{\dagger},
\qquad
\tau(a)=1-\phi(a).
\end{equation*}
Then \(\tau\) is a bijection and reverses the order:
\begin{equation*}
a\ge b
\quad\Longleftrightarrow\quad
\tau(a)\le\tau(b).
\end{equation*}
Extend \(\tau\) coordinate-wise to streams in \(Y^{\mathbb N}\).

By Lemma \ref{L5}, applied to the well ordered subset \(Y^{\dagger}\subset(0,1)\), there exists a monotone social welfare function
\begin{equation*}
W^{\dagger}:(Y^{\dagger})^{\mathbb N}\to\mathbb R
\end{equation*}
satisfying Anonymity, Almost Weak Pareto, and stationarity.
Define
\begin{equation*}
W(x)=-W^{\dagger}(\tau(x)),
\qquad x\in Y^{\mathbb N}.
\end{equation*}

For every finite permutation \(\pi\) of \(\mathbb N\),
\begin{equation*}
\tau(x_{\pi})=\tau(x)_{\pi}.
\end{equation*}
Hence Anonymity of \(W^{\dagger}\) implies Anonymity of \(W\).

If \(x_1=y_1\), then \(\tau(x)_1=\tau(y)_1\) and
\begin{equation*}
\tau(x')=(\tau(x))',
\qquad
\tau(y')=(\tau(y))'.
\end{equation*}
By stationarity of \(W^{\dagger}\), applied after interchanging the two arguments, the equivalent ``less-than-or-equal-to'' formulation is
\begin{equation*}
W^{\dagger}(\tau(x))\le W^{\dagger}(\tau(y))
\quad\Longleftrightarrow\quad
W^{\dagger}(\tau(x'))\le W^{\dagger}(\tau(y')).
\end{equation*}
Multiplying by \(-1\) gives
\begin{equation*}
W(x)\ge W(y)
\quad\Longleftrightarrow\quad
W(x')\ge W(y'),
\end{equation*}
so \(W\) satisfies stationarity.

If \(x\ge y\), then order reversal gives
\begin{equation*}
\tau(y)\ge\tau(x).
\end{equation*}
By monotonicity of \(W^{\dagger}\),
\begin{equation*}
W^{\dagger}(\tau(y))\ge W^{\dagger}(\tau(x)),
\end{equation*}
and hence \(W(x)\ge W(y)\).
Thus \(W\) is monotone.

Finally, suppose \(x\ge y\) and \(x_t>y_t\) for every \(t>\bar N\), for some \(\bar N\in\mathbb N\).
Then
\begin{equation*}
\tau(y)\ge\tau(x)
\end{equation*}
and
\begin{equation*}
\tau(y)_t>\tau(x)_t
\qquad\text{for every }t>\bar N.
\end{equation*}
By AWP for \(W^{\dagger}\),
\begin{equation*}
W^{\dagger}(\tau(y))>W^{\dagger}(\tau(x)).
\end{equation*}
Multiplying by \(-1\) yields \(W(x)>W(y)\).
Hence \(W\) satisfies AWP.
\end{proof}

\subsection{Proof of Lemma \ref{L9}}
\begin{proof}[Proof of Lemma \ref{L9}]
Map the original set into the open unit interval and write
\begin{equation*}
\widehat Y=\phi(Y)\subset(0,1).
\end{equation*}
Because \(\phi\) is an order isomorphism, every non-empty subset of \(\widehat Y\) has either a first element or a last element.
By Lemma \ref{L2}, write
\begin{equation*}
\widehat Y=\widehat Y_1+\widehat Y_2,
\end{equation*}
where \(\widehat Y_1\) is well ordered (possibly empty) and \(\widehat Y_2\) is reverse well ordered (possibly empty), both under the usual order inherited from \((0,1)\).

The \(G_n\)-construction in the proof of Lemma \ref{L5}, based on the fixed enumeration
\begin{equation*}
\mathbb Q\cap(0,1)=\{q_1,q_2,\cdots\},
\end{equation*}
shows that \(\widehat Y_1\) is countable.
For \(\widehat Y_2\), the reflected set
\begin{equation*}
1-\widehat Y_2=\{1-u:u\in\widehat Y_2\}\subset(0,1)
\end{equation*}
is well ordered, so the same \(G_n\)-construction shows that it, and hence \(\widehat Y_2\), is countable.
Therefore \(\widehat Y\) is countable.

For each \(z\in\widehat Y\) and \(u\in\widehat Y^{\mathbb N}\), define the lower and upper cumulative cutoff counts
\begin{equation*}
L_z(u)=\left|\{t\in\mathbb N:u_t\le z\}\right|,
\qquad
U_z(u)=\left|\{t\in\mathbb N:u_t\ge z\}\right|.
\end{equation*}
Both counts take values in \(\mathcal A\).
Since \(\widehat Y\) is countable, introduce the countable non-empty set of tagged tests
\begin{equation*}
\mathcal T=(\{\mathrm L\}\times\widehat Y)
\cup(\{\mathrm U\}\times\widehat Y).
\end{equation*}
For \(\theta\in\mathcal T\), define
\begin{equation*}
T_{\theta}(u)=
\begin{cases}
L_z(u),&\theta=(\mathrm L,z),\\
U_z(u),&\theta=(\mathrm U,z).
\end{cases}
\end{equation*}
Fix a sequence \((\theta_1,\theta_2,\cdots)\) in which every element of \(\mathcal T\) occurs at least once; if \(\mathcal T\) is finite, repeat its elements periodically.
For \(u\in\widehat Y^{\mathbb N}\), define the cumulative profile
\begin{equation*}
r(u)=\big(T_{\theta_1}(u),T_{\theta_2}(u),\cdots\big)
\in\mathcal A^{\mathbb N}.
\end{equation*}
At a coordinate \(i\) with \(\theta_i=(\mathrm L,z)\), a smaller value is better; at a coordinate \(i\) with \(\theta_i=(\mathrm U,z)\), a larger value is better.
Accordingly, define a lexicographic relation \(\succeq_{\mathcal T}\) on \(\mathcal A^{\mathbb N}\) as follows.
For \(p\ne q\), let
\begin{equation*}
k=\min\{i\in\mathbb N:p_i\ne q_i\}.
\end{equation*}
Then \(p\succ_{\mathcal T}q\) if either \(\theta_k=(\mathrm L,z)\) for some \(z\in\widehat Y\) and \(p_k<q_k\), or \(\theta_k=(\mathrm U,z)\) for some \(z\in\widehat Y\) and \(p_k>q_k\).
This relation is complete and transitive: it is the usual lexicographic order after equipping each lower-cutoff coordinate with the reverse of the natural order on \(\mathcal A\) and each upper-cutoff coordinate with the natural order on \(\mathcal A\).

We next establish a real-valued representation of this profile order.
For a finite sequence \(a=(a_1,\cdots,a_k)\in\mathcal A^k\), define \(d(a)\in\mathcal A^{\mathbb N}\) by setting \(d_i(a)=a_i\) for \(i\le k\), and, for \(i>k\),
\begin{equation*}
d_i(a)=
\begin{cases}
\infty,&\theta_i=(\mathrm L,z)\text{ for some }z\in\widehat Y,\\
0,&\theta_i=(\mathrm U,z)\text{ for some }z\in\widehat Y.
\end{cases}
\end{equation*}
Let
\begin{equation*}
\mathcal D_{\mathcal T}
=\{d(a):a\in\mathcal A^k\text{ for some }k\in\mathbb N\}.
\end{equation*}
Since \(\mathcal A\) is countable, \(\mathcal D_{\mathcal T}\) is countable.
Moreover, it is order-separating.
Indeed, if \(p\succ_{\mathcal T}q\) and \(k\) is their first coordinate of difference, take
\begin{equation*}
d=d(p_1,\cdots,p_k).
\end{equation*}
Then \(d\succ_{\mathcal T}q\), because \(d_k=p_k\) and the comparison with \(q\) is decided at coordinate \(k\).
Also \(p\succeq_{\mathcal T}d\): after coordinate \(k\), the profile \(d\) has the worst possible value at every coordinate, namely \(\infty\) for a lower cutoff and \(0\) for an upper cutoff.
Thus
\begin{equation*}
p\succeq_{\mathcal T}d\succ_{\mathcal T}q.
\end{equation*}

Enumerate \(\mathcal D_{\mathcal T}=\{d^1,d^2,\cdots\}\) and define
\begin{equation*}
V_{\mathcal T}(p)=\sum_{s=1}^{\infty}2^{-s}
\mathbf 1_{\{p\succeq_{\mathcal T}d^s\}}.
\end{equation*}
The argument of Lemma \ref{L8} applies verbatim and yields
\begin{equation*}
V_{\mathcal T}(p)\ge V_{\mathcal T}(q)
\quad\Longleftrightarrow\quad
p\succeq_{\mathcal T}q.
\end{equation*}
Define
\begin{equation*}
\widehat W(u)=V_{\mathcal T}(r(u)),
\qquad u\in\widehat Y^{\mathbb N}.
\end{equation*}

An arbitrary permutation of \(\mathbb N\) preserves every lower and upper cumulative cutoff count.
Hence \(\widehat W\) satisfies strong anonymity, and therefore AN.

For stationarity, let \(u_1=v_1=a\), and write
\begin{equation*}
u'=(u_2,u_3,\cdots),
\qquad
v'=(v_2,v_3,\cdots).
\end{equation*}
For a lower cutoff \(L_z\), deleting the common first coordinate leaves both counts unchanged if \(a>z\).  If \(a\le z\), that common coordinate is counted on both sides, so both finite counts are positive; deletion subtracts one from each finite count and leaves \(\infty\) unchanged.
For an upper cutoff \(U_z\), deleting the common first coordinate leaves both counts unchanged if \(a<z\).  If \(a\ge z\), the common coordinate is again counted on both sides, so deletion subtracts one from each positive finite count and leaves \(\infty\) unchanged.
Since this operation preserves equality and strict order, equality and strict inequality at every test coordinate are preserved, and consequently
\begin{equation*}
r(u)\succeq_{\mathcal T}r(v)
\quad\Longleftrightarrow\quad
r(u')\succeq_{\mathcal T}r(v').
\end{equation*}
Therefore \(\widehat W\) satisfies stationarity.

For monotonicity, suppose \(u\ge v\).
Then, for every \(z\in\widehat Y\),
\begin{equation*}
\{t:u_t\le z\}\subseteq\{t:v_t\le z\}
\end{equation*}
and
\begin{equation*}
\{t:v_t\ge z\}\subseteq\{t:u_t\ge z\}.
\end{equation*}
Hence
\begin{equation*}
L_z(u)\le L_z(v)
\qquad\text{and}\qquad
U_z(u)\ge U_z(v).
\end{equation*}
Thus every test coordinate is weakly favorable to \(u\).
If the two cumulative profiles differ, their first difference is favorable to \(u\); otherwise they are equal.
Therefore
\begin{equation*}
r(u)\succeq_{\mathcal T}r(v)
\end{equation*}
and \(\widehat W(u)\ge\widehat W(v)\).
So \(\widehat W\) is monotone.

Finally, suppose \(u\ge v\) and \(u_t>v_t\) for every \(t>\bar N\), for some \(\bar N\in\mathbb N\).
We show that at least one test coordinate is strictly favorable to \(u\).

Suppose first that some tail value of \(u\) lies in \(\widehat Y_1\).
Since \(\widehat Y_1\) is well ordered, the non-empty set
\begin{equation*}
\{u_t:t>\bar N,\ u_t\in\widehat Y_1\}
\end{equation*}
has a first element, say \(a=u_{t^{\ast}}\) for some \(t^{\ast}>\bar N\).
Set \(z=v_{t^{\ast}}\).
Then \(z<a\), and necessarily \(z\in\widehat Y_1\).
By the choice of \(a\), every tail value of \(u\) is strictly greater than \(z\): tail values in \(\widehat Y_1\) are at least \(a\), and every element of \(\widehat Y_2\) lies above every element of \(\widehat Y_1\).
Thus \(L_z(u)\) is finite.
Since \(u\ge v\), one has \(L_z(u)\le L_z(v)\), and \(t^{\ast}\) is counted by \(L_z(v)\) but not by \(L_z(u)\).
Hence
\begin{equation*}
L_z(u)<L_z(v).
\end{equation*}

Suppose next that every tail value of \(u\) lies in \(\widehat Y_2\) and some tail value of \(v\) also lies in \(\widehat Y_2\).
Since \(\widehat Y_2\) is reverse well ordered, the non-empty set
\begin{equation*}
\{v_t:t>\bar N,\ v_t\in\widehat Y_2\}
\end{equation*}
has a last element, say \(b=v_{t^{\ast}}\) for some \(t^{\ast}>\bar N\).
Set \(z=u_{t^{\ast}}\).
Then \(z>b\), and \(z\in\widehat Y_2\).
By the choice of \(b\), every tail value of \(v\) is strictly below \(z\): tail values in \(\widehat Y_2\) are at most \(b\), and every element of \(\widehat Y_1\) lies below every element of \(\widehat Y_2\).
Thus \(U_z(v)\) is finite.
Since \(u\ge v\), one has \(U_z(u)\ge U_z(v)\), and \(t^{\ast}\) is counted by \(U_z(u)\) but not by \(U_z(v)\).
Hence
\begin{equation*}
U_z(u)>U_z(v).
\end{equation*}

Finally, suppose that every tail value of \(u\) lies in \(\widehat Y_2\) and every tail value of \(v\) lies in \(\widehat Y_1\).
Choose any \(t^{\ast}>\bar N\) and set \(z=u_{t^{\ast}}\in\widehat Y_2\).
Every tail value of \(v\) lies below \(z\), so \(U_z(v)\) is finite, whereas \(t^{\ast}\) is counted by \(U_z(u)\) but not by \(U_z(v)\).
Together with \(u\ge v\), this gives
\begin{equation*}
U_z(u)>U_z(v).
\end{equation*}

The three cases are exhaustive.
Hence at least one test coordinate is strictly favorable to \(u\), while all test coordinates are weakly favorable to \(u\) by monotonicity.
Therefore
\begin{equation*}
r(u)\succ_{\mathcal T}r(v),
\end{equation*}
and consequently \(\widehat W(u)>\widehat W(v)\).
Thus \(\widehat W\) satisfies AWP.

Finally, transfer the construction back to \(Y\) by defining
\begin{equation*}
W(x)=\widehat W(\phi(x)),
\qquad x\in Y^{\mathbb N}.
\end{equation*}
As in the proof of Lemma \ref{L5}, the strictly increasing coordinate-wise map \(\phi\) preserves permutations, equality of the first coordinate, coordinatewise weak inequality, and strict tail improvement.
Therefore \(W\) is monotone and satisfies Anonymity, Almost Weak Pareto, and stationarity.
\end{proof}

\clearpage

\bibliographystyle{plainnat}
\bibliography{AWPAnonymity}

\end{document}